\documentclass[isre,nonblindrev]{informs4}

\DoubleSpacedXI

\makeatletter\def\ABSfont{\EGT\linespread{0.9}\selectfont}\makeatother

\usepackage{etoolbox}
\AtBeginEnvironment{table}{\OneAndAHalfSpacedXI\normalsize}

\makeatletter
\renewcommand\@makefntext[1]{\SingleSpacedXI\footnotesize
  \parindent 1em\noindent\ignorespaces{\@makefnmark}\,#1}
\renewcommand\footnoterule{\kern-3\p@ \hrule\@width.4\columnwidth \kern2.6\p@}
\makeatother

\usepackage{natbib}
 \bibpunct[, ]{(}{)}{,}{a}{}{,}%
 \def\bibfont{\DoubleSpacedXI\small\raggedright}%
 \def\bibsep{\smallskipamount}%
\usepackage{graphicx}
\usepackage{tikz}
\usetikzlibrary{arrows.meta, positioning}
\usepackage{multirow}
\usepackage[flushleft]{threeparttable}
\usepackage{subcaption}
\usepackage{mathtools, nccmath}
\usepackage{amssymb}
\usepackage{algorithm}
\usepackage{algpseudocode}
\usepackage{booktabs}
\usepackage{comment}
\usepackage{makecell}
\usepackage{enumitem}
\usepackage{anyfontsize}
\usepackage{array}
\usepackage{placeins}
\usepackage{float}

\TheoremsNumberedThrough
\ECRepeatTheorems
\EquationsNumberedThrough

\MANUSCRIPTNO{}
\renewcommand{\theARTICLETOP}{}

\begin{document}
%%%%%%%%%%%%%%%%

\RUNAUTHOR{Li and Gao}
\RUNTITLE{Why Does Misinformation Propagate Faster?}

\TITLE{Why Does Misinformation Propagate Faster? An Algorithmic Perspective on X}

\ARTICLEAUTHORS{%
\AUTHOR{Pan Li\textsuperscript{*}}
\AFF{Georgia Institute of Technology, \EMAIL{pan.li@scheller.gatech.edu}}
\AUTHOR{Shuang Gao\textsuperscript{*}}
\AFF{Arizona State University, \EMAIL{shuanggao@asu.edu}}
}

\ABSTRACT{%
Misinformation is widely reported to propagate faster on engagement-based platforms \citep{vosoughi2018spread}, yet prior work largely focused on empirical analysis, without identifying a specific algorithmic mechanism that results in this phenomenon. Thanks to the open-sourcing of X's recommendation algorithms in March 2023, we conduct in this paper what is, to our knowledge, the first component-level study of the recommendation algorithm deployed by a social media platform, which examines how each of its components affects misinformation propagation. Specifically, we identify the \emph{engagement fungibility} mechanism embedded in the algorithm, where the final score used for providing tweet recommendations is constructed as a weighted sum of all predicted user activities, such as likes, retweets, replies, and quotes. As a result, a tweet can be repeatedly recommended simply because it is predicted to draw many instant reactions (e.g., likes and retweets), even when it is not expected to draw thoughtful responses (e.g., replies and quotes). Since misinformation typically draws a larger share of its engagement from instant reactions, the engagement fungibility mechanism enables it to receive more recommendation exposure on the X platform than it would receive if thoughtful engagement were required, and to propagate faster.

To empirically validate this mechanism, we re-implement X's recommendation algorithm on the USC X 2024 election corpus, and build a calibrated simulation study to analyze the impact of different scoring rules on the propagation of low-credibility tweets. We find that re-tuning the metric weights has little or even a negative impact on reducing the credibility exposure gap, while those scoring rules that set a precondition of thoughtful engagement for amplification would be able to alleviate the gap significantly. This observation holds across eight groups of 46 robustness checks. It is worth noting that low-credibility tweets do not necessarily receive more overall engagement, and that the advantage they gain from the algorithm arises when the additive score enables the predicted instant reactions to substitute for the thoughtful responses that are unlikely to be drawn. Our diagnosis, therefore, yields a simple and deployable fix, a \emph{reflective-threshold gate} that withholds amplification until a tweet is predicted to draw thoughtful engagement, which we find to reallocate exposure away from low-credibility content at no cost to mainstream exposure and with no loss of engagement.
}


\maketitle
{\renewcommand{\thefootnote}{*}\footnotetext{Both authors contributed equally, and the order of authors was determined by a coin flip.}}
\setcounter{footnote}{0}

% ==============================================================
\section{Introduction}
\label{sec:intro}
% ==============================================================

Misinformation has long been observed to propagate faster, farther, and more broadly than true news on engagement-based platforms. This observation is typically attributed to content-level factors, as misinformation is more novel and more emotionally arousing \citep{vosoughi2018spread,brady2017emotion}, while a series of empirical studies has examined individual-level susceptibility \citep{pennycook2019lazy,pennycook2021psychology}, outrage-driven and low-deliberation sharing \citep{proellochs2023mechanisms,mcloughlin2024misinformation}, the structure of false and true cascades of the same size \citep{juul2021comparing}, and the diffusion of community fact-checked posts \citep{drolsbach2023diffusion}.

However, existing studies rely solely on observed social media records, while the underlying recommendation algorithm is largely overlooked, raising three issues in understanding the misinformation propagation mechanism. First, the observations are heavily influenced by the specific recommender design, so the identified factors might be confounded by the algorithmic allocation. Second, social media platforms may develop vastly different recommendation models for different business objectives and system designs, so the empirical findings cannot be generalized to the whole industry. Third, the propagation step converting engagement into further audience lies inside the ranking system, which is invisible in the observations. As a result, while the empirical literature can document that misinformation propagates differentially, it cannot say whether the difference is inherent to human attention or is created by a particular design choice inside the recommender system, nor can it say how a platform could tackle this problem. Answering these questions requires an algorithmic perspective that looks inside the ranking system to locate the component where the propagation asymmetry originates and to develop an effective solution accordingly.

Fortunately, while the specific design is kept confidential by most social media platforms, the X (formerly known as ``Twitter'') engineering team open-sourced a major portion of the platform's recommendation algorithms in March 2023 \citep{twitter2023algorithm}, which offered an unprecedented look at how the ``For You'' timeline fetches, ranks, and filters content. By probing into the algorithm, we conduct what is, to our knowledge, the first algorithmic diagnosis of the misinformation propagation mechanism.

Our main finding is the \textbf{engagement fungibility} mechanism, i.e., signal-class substitutability in the score-aggregation layer of the recommender system, by which low-credibility content\footnote{Since misinformation has no agreed definition at the level of the individual claim, and source credibility is available at the moment a fresh tweet is recommended, when no verdict on the truth of its claim yet exists, we measure misinformation at the level of the publishing source, labeling tweets by the credibility of the domains they link to.} propagates differentially in engagement-based recommenders. The engagement fungibility mechanism refers to the fact that X adopts the additive formulation of the final score $S = \sum_k w_k\, p_k$ used for tweet recommendations, where $p_k$ is the predicted probability of engagement type $k$ and $w_k$ is its weight. This formulation treats engagement signals from System 1 (fast/reactive) and System 2 (slow/reflective) cognitive processes as substitutable evidence of content quality, so that a single strong fast-class signal (e.g., like or retweet) can drive the aggregate score alone in the absence of any slow-class signal (e.g., reply or quote). This mechanism is also consistent with recent work that applies dual-process theory to recommender systems \citep{kleinberg2024challenge,agarwal2024system2}, which argues that engagement-based ranking tends to privilege System 1 over System 2 cognitive responses. Since low-credibility content tends to spread through emotional, low-deliberation sharing \citep{brady2017emotion,mcloughlin2024misinformation} and draws a larger share of its engagement from instant reactions (as we also empirically validate in Section~\ref{sec:counts-observed} of this paper), the additive score allows its predicted reactions to make up for the thoughtful responses it lacks, so that it receives more recommendation exposure than it would if thoughtful engagement were required, which gives it a propagation advantage.

Validating the engagement fungibility mechanism, however, is a challenging task without full access to X's production logs. Additionally, identifying the impact of each recommendation component on the outcome is intrinsically a counterfactual question, as it asks what would happen to the same content under a modified recommender. We tackle these two challenges by building a calibrated \emph{simulation study}, following the research paradigm advocated by \citet{zhang2020consumption}, where we re-implement the X recommendation algorithm, calibrate the simulator based on the USC X 2024 election corpus so that the counterfactual is anchored to observed behavior, and study the impact of multiple architectural variants. We find that simply changing the weight of these signals does not alleviate, and even widens, the propagation gap between low- and high-credibility content, while those variants privileging slow signals over fast signals, which we formalize as \emph{reflective-gated aggregation}, would effectively reduce the gap. Our findings hold across all 46 robustness checks spanning eight groups.

Building on this diagnosis, we subsequently propose a simple and deployable solution, which reshapes the score-aggregation layer with a \emph{reflective-threshold gate}. The gate scales down the additive score of a tweet when its predicted slow-class engagement falls below a threshold, so that a tweet needs to be expected to draw a sufficient amount of thoughtful responses before its one-click reactions can fully amplify it. We demonstrate through further experiments that this solution would be able to reallocate exposure away from low-credibility content at no cost to overall engagement, and we also show how its benefit and cost change with the threshold, so that a platform can choose the optimal threshold that fits its own objectives best.

In this paper, we make the following contributions. First, we identify engagement fungibility, the substitutability of System 1 and System 2 engagement signals in the score-aggregation layer, as the algorithmic mechanism of misinformation propagation. Second, we develop a simulation framework for analyzing the counterfactual outcomes of component-level recommender system design, where we re-implement the X recommendation algorithm, embed it in an agent-based cascade simulator, and calibrate it against the USC X 2024 election corpus. Third, we propose a simple and deployable solution, the reflective-threshold gate, which withholds amplification until predicted slow-class engagement clears a platform-tunable threshold.

% ==============================================================
\section{Related Work}
\label{sec:related}
% ==============================================================

\subsection{Misinformation Research in Information Systems}
\label{sec:related-misinfo}

Beyond the empirical evidence introduced in Section~\ref{sec:intro}, misinformation has been studied in the IS literature at both the user and interface levels. For example, researchers have characterized community-driven rumoring on Twitter during crises \citep{oh2013community}, analyzed how source ratings change whether users believe and would share an article \citep{kim2019says}, and studied the rating designs in platform interventions \citep{kim2019combating}. In addition, neurophysiological evidence shows that social media users largely fail to deliberate on headlines that align with their prior beliefs \citep{moravec2019fake}. It is worth noting that these studies focus on how users judge the content they have been shown, while little attention has been paid to how the content gets shown in the first place.

Meanwhile, an economics-oriented research stream formally models the platform's incentives and optimal moderation policies. A model of a stylized sharing network shows that an engagement-maximizing platform would tilt its algorithm toward filter bubbles that amplify misinformation \citep{acemoglu2024model}, while a parallel stream studies the platform's optimal inspection and signaling policies against fake news \citep{papanastasiou2020fake,candogan2020optimal}. These models treat the algorithm as an abstract policy, so they cannot identify which part of the recommender system is responsible for misinformation propagation, nor what the platform should modify accordingly.

\subsection{Algorithmic Amplification on Social Media}
\label{sec:related-platform}

Beyond misinformation, a broad literature studies how engagement-based recommendation algorithms amplify certain types of content across platforms. On YouTube, audits trace radicalization pathways through recommended channels \citep{ribeiro2020auditing}, and panel data show how users consume radical content on the platform \citep{hosseinmardi2021examining}. On Facebook, misinformation sources draw disproportionate engagement per follower \citep{edelson2021understanding}, algorithmically curated news exposure shows asymmetric ideological segregation during the 2020 U.S.\ election \citep{gonzalez2023asymmetric}, and reducing like-minded exposure in the feed does not measurably change users' attitudes \citep{nyhan2023like}. In addition, the aggregate effect of the feed algorithm on Facebook and Instagram has been measured by switching real users to a reverse-chronological feed during the 2020 U.S.\ election \citep{guess2023algorithms}. A field experiment further shows that algorithmically mediated news exposure affects affective polarization \citep{levy2021social}.

On X, the open-sourcing of the recommendation algorithm in March 2023 \citep{twitter2023algorithm} made the platform's ranking system publicly auditable. Even before the open-sourcing, studies found that the algorithmic timeline amplifies right-leaning political content \citep{huszar2022algorithmic} and curates news differently from a reverse-chronological baseline \citep{bandy2021curating}. More recent audits document the amplification of low-credibility and political content \citep{corsi2023evaluating,ye2025auditing}, find that the timeline favors emotionally charged and toxic content \citep{bouchaud2023crowdsourced}, and show that it amplifies divisive content that users say they do not want \citep{milli2025engagement}. It is worth noting that these studies rely only on empirical analyses, while the internal algorithmic components remain unexplored.

\subsection{Recommender Systems and Simulation-Based Research}
\label{sec:related-dualproc}
\label{sec:related-sim}

Within the broader recommender systems literature in IS, it is well established that recommender design causally shapes downstream outcomes. For example, it is shown that recommendations act as anchors that pull consumers' own preference ratings toward the system's predictions \citep{adomavicius2013do}, personalization drives consumption toward commonality in a field setting \citep{hosanagar2014will}, recommender design shifts sales diversity in a randomized field experiment \citep{lee2019how}, and network structure governs the diffusion of user-generated content \citep{susarla2012social}. This stream of research treats the recommender as a whole unit of analysis, without separating the components of its internal design.

Methodologically, agent-based simulation has an established role in IS research for studying system-level dynamics under interventions that are infeasible or confounded in the field \citep{davis2007developing,harrison2007simulation,miller2007complex}. For example, simulation has been used to show how recommender systems shape sales diversity over time \citep{fleder2009blockbuster}, and an agent-based framework has been developed to understand the longitudinal performance dynamics of recommender systems \citep{zhang2020consumption}. Closer to our analysis, recent simulation studies compare how different recommender families spread misinformation \citep{pathak2023understanding,fernandez2024analysing} and quantify how vulnerable information quality is to manipulation \citep{truong2024vulnerabilities}, though these studies compare whole algorithms with simple baselines and do not vary a single component inside one algorithm.

\subsection{Research Gap}
\label{sec:related-gap}

From the above literature, we identify three research gaps. First, the misinformation literature documents that false news spreads further than true news, yet the recommendation algorithm that decides content exposure remains an unobserved factor. Second, the amplification literature establishes that engagement-based ranking systematically amplifies certain types of content, while neither audits nor field experiments can intervene on a single component inside the algorithm. Third, recommender systems research, along with simulation studies, shows that the ranking layer causally shapes what gets consumed and diffused, without probing a deployed architecture to identify which internal component drives the outcome. This paper addresses these gaps by reconstructing the open-sourced X recommendation algorithm, intervening on its score-aggregation layer in a calibrated simulation, and identifying the component-level mechanism of misinformation propagation. To our knowledge, this paper is the first study to locate the mechanism of misinformation propagation inside X's algorithm, as positioned in Table~\ref{tab:positioning}.

\begin{table}[h]
\centering
\small
\caption{Positioning of this study relative to the reviewed literature.}
\label{tab:positioning}
\begin{tabular}{>{\raggedright\arraybackslash}p{0.29\linewidth} >{\raggedright\arraybackslash}p{0.19\linewidth} >{\raggedright\arraybackslash}p{0.25\linewidth} c}
\toprule
Research stream & Level of analysis & Method & Inside algorithm \\
\midrule
Misinformation research (\S\ref{sec:related-misinfo}) & User and interface & Lab and survey experiments, analytical models & No \\
Algorithmic amplification (\S\ref{sec:related-platform}) & Platform output & Audits, field experiments & No \\
Recommender and simulation (\S\ref{sec:related-sim}) & Whole algorithm & Field experiments, simulation & No \\
\textbf{This paper} & Algorithmic component & Calibrated simulation of the open-sourced algorithm & Yes \\
\bottomrule
\end{tabular}
\end{table}

% ==============================================================
\section{Preliminaries: X's Algorithm and Theoretical Foundation}
\label{sec:theory}
% ==============================================================

\subsection{The X Recommendation Algorithm}
\label{sec:pipeline}

Thanks to the open-sourcing of the recommendation algorithm in March 2023 \citep{twitter2023algorithm}, the internal structure of X's \emph{For You} timeline was made public, which lets us study its deployed pipeline, as we summarize in Figure~\ref{fig:pipeline}. For each timeline request, a candidate-sourcing stage assembles roughly 1{,}500 candidate tweets, drawing about half from accounts the user follows and the remainder from out-of-network sources through community-embedding and graph-based services. A logistic regression model, the Light Ranker, then narrows this pool, and the resulting candidates are fed into the \emph{Heavy Ranker}, a MaskNet neural network \citep{wang2021masknet} that takes several thousand user, author, tweet, and engagement-history features as input and produces for every candidate a vector of predicted engagement probabilities, covering replies, retweets, likes, profile clicks, video views, and negative feedback. Finally, the \emph{score-aggregation layer} combines this vector into a single score through a weighted sum ($13.5$ for reply, $1.0$ for retweet, $0.5$ for like, etc.), and the candidates are ranked by this score.\footnote{In January 2026, X open-sourced its recommender a second time \citep{xai2026algorithm}, updating its prediction model while retaining the additive score-aggregation form, and the ranking weights it published in August 2026 cut the reply weight from $13.5$ to $5.0$. Our corpus was generated under the 2023-era algorithm, which is therefore the architecture we reconstruct. We further evaluate the 2026 re-weighting inside our framework in Section~\ref{sec:sign-asymmetry}, and Check~8 in Appendix~\ref{app:robust} repeats our analysis with a transformer-based Heavy Ranker, which is used in the 2026 system.}

Among these pipeline components, we focus on the Heavy Ranker and its score-aggregation layer in this paper, for the following two reasons. First, the upstream stages mainly control the composition of the candidate pool, while the ordering that determines how much exposure each candidate receives is decided mainly by the Heavy Ranker's scores. Second, within the Heavy Ranker, the engagement objectives and the score-aggregation layer play different roles, as the per-objective predictions estimate how users would react to a candidate, which is a forecasting task, while the aggregation layer sets how much each predicted reaction is worth in the final score, which is a judgment made by the platform. Therefore, the aggregation layer is a design choice that a platform can easily change, which is why we focus our analysis on it.

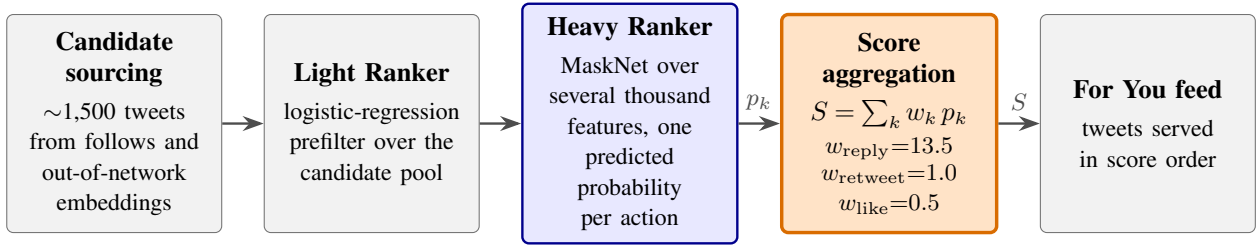
\begin{figure}[t]
\centering
% Node text is single-spaced; the body's double spacing would otherwise apply.
\resizebox{\linewidth}{!}{\SingleSpacedXI%
\begin{tikzpicture}[
  node distance=0.55cm,
  stage/.style={draw=black!55, rounded corners=2.5pt, align=center, minimum height=2.9cm, text width=2.5cm, inner sep=5pt, fill=black!5, font=\small},
  focus/.style={stage, fill=blue!9, draw=blue!55!black, line width=0.9pt},
  target/.style={stage, fill=orange!22, draw=orange!85!black, line width=1.3pt},
  flow/.style={-{Stealth[length=2.8mm, width=2mm]}, line width=0.9pt, black!70},
  lbl/.style={font=\footnotesize\itshape, black!70, align=center},
  tagnote/.style={font=\footnotesize\itshape, align=center}
]
\node[stage] (src) {\textbf{Candidate sourcing}\\[3pt]\footnotesize $\sim$1{,}500 tweets from follows and out-of-network embeddings};
\node[stage, right=of src] (light) {\textbf{Light Ranker}\\[3pt]\footnotesize logistic-regression prefilter over the candidate pool};
\node[focus, right=of light] (heavy) {\textbf{Heavy Ranker}\\[3pt]\footnotesize MaskNet over several thousand features, one predicted probability per action};
\node[target, right=of heavy] (agg) {\textbf{Score}\\\textbf{aggregation}\\[3pt]$S=\sum_k w_k\,p_k$\\[2pt]\footnotesize $w_{\mathrm{reply}}{=}13.5$\\ $w_{\mathrm{retweet}}{=}1.0$\\ $w_{\mathrm{like}}{=}0.5$};
\node[stage, right=of agg] (serve) {\textbf{For You feed}\\[3pt]\footnotesize tweets served in score order};

\draw[flow] (src) -- (light);
\draw[flow] (light) -- (heavy);
\draw[flow] (heavy) -- node[lbl, above=1pt] {$p_k$} (agg);
\draw[flow] (agg) -- node[lbl, above=1pt] {$S$} (serve);

\end{tikzpicture}}
\caption{The X \emph{For You} recommendation pipeline (March 2023), with the two stages this paper studies highlighted in color.}
\label{fig:pipeline}
\end{figure}

\subsection{Dual-Process Theory and the Cognitive Partition of Engagement}
\label{sec:dualprocess}

Our theoretical foundation is \emph{dual-process theory} \citep{kahneman2011thinking,stanovich2000individual,evans2013dual}, which classifies cognitive processes by the effort they demand, separating fast, automatic System 1 processes from slow, effortful System 2 processes. In IS, the theory has informed user-level misinformation interventions \citep{moravec2020appealing}, and it has recently been applied to recommender systems at the level of user welfare \citep{kleinberg2024challenge,agarwal2024system2}. We draw on this framing to classify engagement actions by the cognitive process that produces them, which we call the \emph{cognitive-effort partition} and summarize in Table~\ref{tab:partition}.

\begin{table}[h]
\centering
\caption{Cognitive-effort partition of engagement actions.}
\label{tab:partition}
\begin{tabular}{lll}
\toprule
Action & Cognitive effort & Dual-process class \\
\midrule
Like         & one click, no reading                          & \textbf{fast (System 1)} \\
Retweet      & one click, often no reading, in-group signaling& \textbf{fast (System 1)} \\
Reply        & read $+$ compose text $+$ defend a position    & \textbf{slow (System 2)} \\
Quote-tweet  & read $+$ compose $+$ frame for own audience    & \textbf{slow (System 2)} \\
\bottomrule
\end{tabular}
\end{table}

It is worth noting that X's published $w_{\mathrm{reply}} = 13.5$ (vs.\ $w_{\mathrm{retweet}} = 1.0$) is consistent with the cognitive-effort classification, as the production ranker weights effortful signals more highly, presumably because they are more informative of content quality. However, whatever values the weights take, the additive form makes every predicted action exchangeable for every other at a fixed rate, so a tweet can compensate for a low predicted reply probability with high predicted one-click reactions. The misinformation literature also locates the sharing of false content in low-deliberation, emotion-driven processes \citep{pennycook2019lazy,pennycook2021psychology,mcloughlin2024misinformation,brady2017emotion,vosoughi2018spread}, so that the partition predicts that content from low-credibility sources draws a disproportionate share of its engagement from the fast class. In our corpus, among tweets that drew the same total engagement, replies make up $33.0\%$ of the total engagement that a low-credibility tweet draws on average, against $46.4\%$ for a high-credibility tweet. Therefore, what our mechanism requires is the cognitive-effort partition of engagement into fast and slow classes, and not any particular set of weight values, as discussed in the next section.

% ==============================================================
\section{The Engagement Fungibility Mechanism}
\label{sec:fungibility}
% ==============================================================

\subsection{Definition of the Mechanism}

Our identified mechanism is \textbf{engagement fungibility}, i.e., signal-class substitutability in the score-aggregation layer. Under this mechanism, the score-aggregation layer treats the engagement signals from cognitively distinct processes, fast (System 1) and slow (System 2), as substitutable evidence of content quality, so that any single high signal can drive the score regardless of which cognitive class produced it, which gives an advantage in exposure to content that draws a larger share of its engagement from instant reactions (Section~\ref{sec:dualprocess}). In economic terms, the linear form prices the two classes as perfect substitutes, and a like is fungible with a reply at a fixed exchange rate set by the weights.

Formally, let $\mathbf{p} = (p_1, \ldots, p_K) \in [0,1]^K$ denote the ranker's predicted engagement probabilities, partitioned into a slow class $\mathcal{S}$ and a fast class $\mathcal{F}$ by the cognitive-effort partition of Section~\ref{sec:dualprocess}. An aggregation rule, which we also call a scoring rule, is a differentiable function $\Phi: [0,1]^K \rightarrow \mathbb{R}$ that is increasing in every slow-class prediction and maps the predictions to the final score $S = \Phi(\mathbf{p})$, and we write $S_{\mathrm{slow}} = \sum_{i \in \mathcal{S}} w_i\, p_i$ and $S_{\mathrm{fast}} = \sum_{j \in \mathcal{F}} w_j\, p_j$ for the class aggregates under the weights of the additive rule. The marginal rate of substitution $\rho_{ij}(\mathbf{p}) = \frac{\partial \Phi / \partial p_j}{\partial \Phi / \partial p_i}$, for $j \in \mathcal{F}$ and $i \in \mathcal{S}$, measures how many units of predicted slow-class engagement one unit of predicted fast-class engagement can replace at the margin.

\begin{definition}[Engagement fungibility]
\label{def:fungibility}
An aggregation rule $\Phi$ exhibits \emph{engagement fungibility} if $\rho_{ij}(\mathbf{p}) = \rho_{ij} > 0$ for all $\mathbf{p} \in [0,1]^K$, every $j \in \mathcal{F}$, and every $i \in \mathcal{S}$.
\end{definition}

The additive rule $\Phi_{\mathrm{add}}(\mathbf{p}) = \sum_k w_k\, p_k$ exhibits the engagement fungibility of Definition~\ref{def:fungibility} with exchange rate $\rho_{ij} = w_j / w_i$ (e.g., $w_{\mathrm{like}} / w_{\mathrm{reply}} = 0.5/13.5 = 1/27$), regardless of how little slow-class engagement a tweet is predicted to draw. The additive rule is, in fact, canonical for the entire fungible class, as we show below in Theorem~\ref{thm:canonical}. Since the recommender selects candidates by their score ordering, we call two aggregation rules \emph{ordinally equivalent} if they induce the same ordering of candidates on $[0,1]^K$.

\begin{theorem}[Canonical representation of fungible rules]
\label{thm:canonical}
An aggregation rule $\Phi$ exhibits engagement fungibility if and only if $\Phi(\mathbf{p}) = g\bigl(\sum_k c_k\, p_k\bigr)$ for a differentiable function $g$ with $g' > 0$ and positive weights with $\sum_k c_k = 1$, and the weight vector $\mathbf{c}$ is unique. Every fungible rule is therefore ordinally equivalent to one and only one normalized additive rule.
\end{theorem}

\proof{Proof.}
Constancy of the cross-class rates forces every pairwise ratio of marginal values to be constant, since $\partial \Phi/\partial p_{i'} \,\big/\, \partial \Phi/\partial p_i = \rho_{ij}/\rho_{i'j}$ for $i, i' \in \mathcal{S}$, $j \in \mathcal{F}$, and likewise within $\mathcal{F}$. The gradient of $\Phi$ is then everywhere proportional to a fixed positive vector $\mathbf{c}$, normalized to $\sum_k c_k = 1$, so the level sets are parallel hyperplanes and $\Phi = g(\mathbf{c} \cdot \mathbf{p})$ with $g$ increasing. The converse holds since any increasing transform of a weighted sum has $\rho_{ij} = c_j/c_i$. For uniqueness, if normalized $\mathbf{c}$ and $\mathbf{c}'$ induce the same ordering on $[0,1]^K$, their level sets coincide, so $\mathbf{c}' = \lambda\, \mathbf{c}$ with $\lambda > 0$, and normalization gives $\lambda = 1$. \Halmos
\endproof

Theorem~\ref{thm:canonical} leads to the following two implications for our study. First, it shows that a fungible rule cannot give content different exposure depending on whether its predicted engagement is reflective or reactive. Second, it shows that each fungible ordering corresponds to one and only one set of normalized weights, so an analysis that covers the space of additive weights applies to every fungible rule. This motivates our design in which fast-class predictions carry value only alongside slow-class validation, as discussed next.

\subsection{Potential Solution and Hypotheses}
\label{sec:forms}
\label{sec:propositions}

By Theorem~\ref{thm:canonical}, avoiding engagement fungibility requires an aggregation rule whose exchange rate between fast- and slow-class engagement changes with the predictions. The dual-process theory introduced in Section~\ref{sec:dualprocess} motivates our potential solution, where we treat predicted slow-class engagement as the evidence of quality and let predicted fast-class engagement amplify only what that evidence supports. We formalize the solution as \emph{reflective-gated aggregation} in Definition~\ref{def:gated} through a bound on the score when predicted slow-class engagement is absent.

\begin{definition}[Reflective-gated aggregation]
\label{def:gated}
An aggregation rule $\Phi$ is \emph{reflective-gated} if it does not exhibit engagement fungibility and there exist $\gamma_0 \in [0,1)$ and $\beta_0 \geq 0$ such that $\Phi(\mathbf{p}) \leq \gamma_0\, \Phi_{\mathrm{add}}(\mathbf{p}) + \beta_0$ whenever $S_{\mathrm{slow}}(\mathbf{p}) = 0$.
\end{definition}

Reflective-gated aggregation can be instantiated in various formulations, including the following three functional forms, which we label F1, F2, and F3.

\textbf{(F1) Slow-gates-fast multiplicative.} $S = S_{\mathrm{slow}} \cdot \bigl(1 + \alpha \cdot S_{\mathrm{fast}}\bigr)$, with $\alpha > 0$ controlling how much fast-class engagement can amplify slow-validated content. Since $\partial \Phi_{F1}/\partial p_j = \alpha\, w_j\, S_{\mathrm{slow}}$ for $j \in \mathcal{F}$, the substitution rate vanishes as $S_{\mathrm{slow}} \rightarrow 0$, and the form satisfies Definition~\ref{def:gated} with $\gamma_0 = \beta_0 = 0$.

\textbf{(F2) Ratio correction.} $S = \frac{S_{\mathrm{additive}}}{1 + \alpha \cdot S_{\mathrm{fast}} / (S_{\mathrm{slow}} + \epsilon)}$, where $\epsilon > 0$ is a small constant, which down-weights the additive score by the ratio of predicted fast- to slow-class engagement, so the penalty grows with the fast-class predictions that would otherwise drive the additive score. At $S_{\mathrm{slow}} = 0$ the score equals $S_{\mathrm{fast}}\, \epsilon / (\epsilon + \alpha\, S_{\mathrm{fast}}) \leq \epsilon/\alpha$, so it satisfies Definition~\ref{def:gated} with $\gamma_0 = 0$ and $\beta_0 = \epsilon/\alpha$.

\textbf{(F3) Reflective-threshold sigmoid.} $S = S_{\mathrm{additive}} \cdot \sigma\!\left(\frac{S_{\mathrm{slow}} - \theta}{\mathrm{scale}}\right)$, where $\sigma$ is the logistic function and $\mathrm{scale} > 0$ sets how steeply the gate moves from closed to open, which scales the additive score down smoothly as $S_{\mathrm{slow}}$ falls below the threshold $\theta$. At $S_{\mathrm{slow}} = 0$ the score is $S_{\mathrm{fast}}\, \sigma(-\theta/\mathrm{scale})$, so it satisfies Definition~\ref{def:gated} with $\gamma_0 = \sigma(-\theta/\mathrm{scale})$ and $\beta_0 = 0$.

We denote $c(x) \in \{\mathrm{low}, \mathrm{high}\}$ as the credibility class of a tweet $x$ that starts a simulated cascade, which we call a seed tweet, and $Y(x; \Phi)$ as a propagation outcome for $x$ when the score-aggregation rule is $\Phi$. The \emph{propagation gap} under $\Phi$ is $\Delta(\Phi) \;=\; \mathbb{E}\bigl[\,Y(x;\Phi) \mid c(x)=\mathrm{low}\,\bigr] \;-\; \mathbb{E}\bigl[\,Y(x;\Phi) \mid c(x)=\mathrm{high}\,\bigr],$ and $\mathcal{C}(\Phi) = \Delta(\Phi) - \Delta(\Phi_{\mathrm{add}})$ denotes the change in this gap under an alternative rule relative to the additive baseline, so that $\mathcal{C}(\Phi) < 0$ means $\Phi$ narrows the gap. In addition, we write $R(x) = \sum_{j \in \mathcal{F}} w_j\, p_j(x) \big/ \sum_{i \in \mathcal{S}} w_i\, p_i(x)$ for the predicted reactive-to-reflective ratio of seed $x$. We now state the following three hypotheses to demonstrate the effectiveness of reflective-gated aggregation.

\begin{hypothesis}[The aggregation form narrows the gap]
\label{prop:closure}
Reflective-gated aggregation narrows the propagation gap between low- and high-credibility content relative to the additive baseline.
\end{hypothesis}

\begin{hypothesis}[Re-tuning the weights does not help]
\label{prop:parameter}
Re-tuning the weights of the additive rule, which keeps the aggregation fungible, does not narrow the propagation gap.
\end{hypothesis}

\begin{hypothesis}[The gate acts on the engagement mix]
\label{prop:targeting}
A reflective-gated rule demotes content in proportion to how far its predicted engagement tilts toward the fast class, and the propagation gap narrows because low-credibility content tilts further on average.
\end{hypothesis}

Together, Hypotheses~\ref{prop:closure} and~\ref{prop:parameter} assert that narrowing the gap is a property of the functional form, not of the weight values, which are held fixed across the reflective-gated forms. The tilt in Hypothesis~\ref{prop:targeting} is measured by the predicted ratio $R(x)$ at a given level of predicted slow-class engagement, so the hypothesis attributes demotion to the composition of a seed's predicted profile and not to its overall engagement level. Next, we will build a comprehensive simulation framework to test these hypotheses.

% ==============================================================
\section{Data and Simulation Framework}
\label{sec:data}
% ==============================================================

\subsection{USC X 2024 Election Corpus}
\label{sec:data-usc}
\label{sec:data-limits}

We select the USC X 2024 U.S.\ Election Dataset \citep{balasubramanian2024usc} for our study, which contains tweets collected via election-related keyword filtering between May 1 and November 30, 2024. We choose this dataset for the following two reasons. First, a U.S.\ presidential election is a peak period for misinformation circulation, which gives the corpus low-credibility content at a scale sufficient for class-level comparison. Second, its cascades were generated while the open-sourced 2023 X algorithm was in production, so the observed engagement patterns and the ranker we reconstruct come from the same algorithmic era. The data preprocessing details are described in Appendix~\ref{app:data}.

Table~\ref{tab:descriptives} reports descriptive statistics for the records in our study. We find that election discourse is predominantly conversational, as replies account for around two-thirds of the tweets and original posts for nearly one-quarter, while roughly one tweet in ten carries an embedded URL. We also find that engagement is heavily concentrated, since the median tweet draws no replies, retweets, or quote-tweets at all and receives $17$ impressions, while the tweets in the top percentile draw at least $20$ replies and $206$ likes, which is the pattern that the event model of Section~\ref{sec:framework-event} is built to reproduce.

\begin{table}[h]
\centering
\small
\caption{Descriptive statistics of the corpus (May 1 to November 30, 2024).}
\label{tab:descriptives}
\begin{tabular}{lr@{\hspace{2.4em}}lrr}
\toprule
 & & & Mean & Median \\
\midrule
\textit{Coverage} & & \textit{Engagement per tweet} & & \\
\quad Tweets                & $40{,}302{,}075$ & \quad Replies             & $3.81$    & $0$ \\
\quad Unique authors        & $4{,}292{,}568$  & \quad Retweets            & $8.31$    & $0$ \\
\quad Unique conversations  & $17{,}433{,}960$ & \quad Likes               & $39.35$   & $0$ \\
                            &                  & \quad Quote-tweets        & $0.63$    & $0$ \\
\textit{Composition (\% of tweets)} & &         \quad Impressions          & $2{,}108$ & $17$ \\
\quad Replies               & $66.1$           &                           &            & \\
\quad Original posts        & $23.4$           & \textit{Author characteristics} & & \\
\quad Quote-tweets          & $11.0$           & \quad Followers           & $4{,}451$  & $194$ \\
\quad With an embedded URL  & $9.9$            & \quad Accounts followed   & $934$      & $278$ \\
\quad English               & $86.8$           & \quad Tweets posted       & $15{,}011$ & $2{,}947$ \\
\quad Paid verification     & $22.9$           & \quad Favorites given     & $25{,}225$ & $4{,}400$ \\
\quad Credibility-labeled   & $2.2$            & \quad Lists appearing on  & $21$       & $0$ \\
\bottomrule
\end{tabular}
\end{table}

\subsection{Labeling of Credibility}
\label{sec:data-labels}

We label tweets by matching the publishing domains of their embedded URLs against two lists. The low-credibility list is Iffy+ \citep{funke2026iffy}, which aggregates fake-news and misinformation lists from BuzzFeed, FactCheck.org, PolitiFact, Wikipedia, and the Conspiracy-Pseudoscience and Questionable Sources categories of Media Bias/Fact Check (MBFC), whose factual reporting is rated Very Low, Low, or Mixed. The high-credibility list is a hand-curated list of mainstream news outlets, most of which MBFC rates ``High'' or ``Very High'' in factual reporting \citep{mbfc2026methodology}, balanced across the partisan spectrum and covering international as well as U.S.\ outlets, and MBFC's ratings agree closely with other expert and crowd ratings of news-domain quality \citep{lin2023high}. A tweet is labeled low-credibility if all matched domains are in Iffy+ and high-credibility if all are in the high-credibility list. A tweet matching both lists is labeled mixed, and one matching neither remains unlabeled. This yields labels for $\sim$2.2\% of the records in the corpus, with $0.82\%$ of all records labeled low-credibility and $1.38\%$ labeled high-credibility. The labels thus measure the credibility of the publishing source, which is the construct a scoring-time intervention can act on, since no verdict on whether a tweet's claim is true exists when the tweet is first recommended.

% ==============================================================
\subsection{Simulation Framework Overview}
\label{sec:framework}

Our research question is counterfactual, as it asks what would happen to the same content under a different aggregation rule. Since neither algorithmic audits \citep{huszar2022algorithmic,bandy2021curating,corsi2023evaluating,bouchaud2023crowdsourced,ye2025auditing} nor field experiments \citep{guess2023algorithms,nyhan2023like,levy2021social} can answer this question, we adopt the simulation paradigm that IS research has established for interventions that are infeasible or confounded in the field \citep{davis2007developing,harrison2007simulation,miller2007complex,fleder2009blockbuster,zhang2020consumption}. It is worth noting that existing simulation studies on this topic compare different recommenders \citep{pathak2023understanding,fernandez2024analysing,truong2024vulnerabilities}, while our design varies a single layer within the same recommender.

Our simulation system comprises three components, which mirror the stages of the deployed pipeline in Figure~\ref{fig:pipeline}. The first is a \emph{seed population} sampled from observed tweets, each treated as the root of a fresh cascade. The second is a \emph{user population} sampled from the observed author-feature distribution. The third is the \emph{ranker and score-aggregation layer}, which maps each seed to a score under a particular scoring rule and thereby determines its algorithmic exposure. All three components are instantiated from the USC corpus, as we calibrate agent-based models against empirical data \citep{rand2011agent}, and Section~\ref{sec:results-stage1} verifies that the calibrated system reproduces the observed cascade patterns. Table~\ref{tab:framework-components} summarizes the three components, describing for each how it is built and which attributes each of its elements carries.

Our design identifies the architectural effect by holding everything but the scoring rule fixed, including the trained ranker and its learned parameters, the seed and user populations, the exposure mapping, and all calibrated parameters. Additionally, user features are static, and the ranker is not retrained within the simulation, so that our estimates capture the one-step effect of the design change \citep{chaney2018algorithmic,zhang2020consumption}. We model cascades through the total count of each engagement type, and its hourly schedule, as we will explain next.

\begin{table}[t]
\centering
\small
\caption{Overview of the three simulated-system components.}
\label{tab:framework-components}
\begin{tabular}{>{\raggedright\arraybackslash}p{0.15\linewidth} >{\raggedright\arraybackslash}p{0.42\linewidth} >{\raggedright\arraybackslash}p{0.33\linewidth}}
\toprule
\textbf{Component} & \textbf{Population-level modeling} & \textbf{Per-agent state (static)} \\
\midrule
User pool & 50{,}000 users sampled from observed USC authors, stratified by follower-count quintile $\times$ paid verification. & Follower/friend/status/favorite/list counts and the paid-verification flag \\
\midrule
Seed population & Sampled from observed tweets, stratified by credibility label for the hypothesis test, each seed treated as a fresh cascade root. & Posting time, text, embedded URLs, and engagement counts (training targets only) \\
\midrule
Ranker $+$ aggregation layer & Parallel MaskNet \citep{wang2021masknet} with four engagement objectives, published weights for replies, retweets, and likes, and the scoring rule as the only component varied. & Per-objective probabilities, aggregate and relative scores, exposure allocation, and activity indicator \\
\bottomrule
\end{tabular}
\end{table}

\subsection{Ranker, Score Aggregation, and Exposure}
\label{sec:framework-ranker}
\label{sec:framework-exposure}

We implement the ranker as a parallel MaskNet \citep{wang2021masknet}, following the architecture of the open-sourced Heavy Ranker. Our ranker passes 18 account, tweet, and time-of-day features (listed in Appendix~\ref{app:data}) through three instance-guided masked blocks and a shared output layer, and outputs four engagement predictions, for replies, retweets, likes, and quotes. For each seed, the trained ranker produces a per-objective engagement probability $p_k = \sigma\bigl(f_k(\mathbf{x})\bigr)$ for $k \in \{\mathrm{reply},\,\mathrm{retweet},\,\mathrm{like},\,\mathrm{quote}\}$, where $\mathbf{x}$ is the seed's standardized feature vector, $f_k$ is the logit of objective $k$, and $\sigma$ is the logistic function. The score-aggregation layer then collapses the four predictions into a single score that determines exposure:
\begin{equation}
S \;=\; \sum_k w_k\, p_k,
\label{eq:framework-additive}
\end{equation}
where $w_k$ is the weight of objective $k$, taken from the published Heavy Ranker weights for replies ($13.5$), retweets ($1.0$), and likes ($0.5$), and set to $2.0$ for quotes, which have no counterpart among the published objectives. The reflective-gated forms differ from the additive rule only in this layer, as each rule replaces Eq.~\eqref{eq:framework-additive} with one of the three forms that we already introduced in Section~\ref{sec:forms}.

The aggregate score determines how widely the algorithm circulates a seed tweet. To make this concrete, we normalize each score by the median score of the seed population, $S_{\mathrm{rel}} = S / \operatorname{med}(S)$, and convert the relative score into a per-cascade exposure allocation $\mathit{E}_i = \operatorname{clip}\bigl(\beta \cdot S_{\mathrm{rel},i},\, E_{\min},\, E_{\max}\bigr)$, where $\beta$ is the calibrated mean number of exposures for a median-scored cascade, and the bounds $E_{\min} = 0$ and $E_{\max} = 50{,}000$ keep extreme scores within a plausible range. The allocation serves as the simulator's measure of exposure, since the corpus does not record which users are actually shown a tweet. The allocation is then spread over the 24 hours after posting by a schedule that combines a decay over time with the daily pattern of user activity, so that a seed posted at hour $h_i$ receives the following exposure in hour $t$ after posting:
\begin{equation}
q_{i,t} \;=\; \mathit{E}_i \cdot \frac{\exp(-t/\tau)\; d_{(h_i+t)\bmod 24}}{\sum_{s=0}^{23} \exp(-s/\tau)\; d_{(h_i+s)\bmod 24}}, \qquad t = 0, 1, \ldots, 23,
\label{eq:framework-schedule}
\end{equation}
where $\tau = 6$ hours is the decay constant and $d_h$ measures how busy hour $h$ of the day is relative to an average hour, which we take from the times at which tweets were posted in our corpus.

\subsection{Engagement Event Model}
\label{sec:framework-event}

The event model translates a cascade's exposure into engagement counts in two stages to reproduce two empirical regularities of engagement on X, namely that most tweets receive no engagement at all, and that the engagement concentrates heavily on a small share of tweets. In the first stage, $a_i \sim \mathrm{Bernoulli}(\pi_{\mathrm{active}})$ determines whether the seed tweet $i$ attracts any engagement at all, where $\pi_{\mathrm{active}}$ is the calibrated share of active cascades, which captures the large share of tweets that circulate without response. In the second stage, conditional on activity, engagement accrues over the 24-hour horizon according to
\begin{align}
\mu_{i,k} &\;=\; \mathit{E}_i \cdot p_{i,k}, \label{eq:framework-mu}\\
\lambda_{i,k} &\;\sim\; \mathrm{Gamma}\!\bigl(r,\; \mu_{i,k}/r\bigr), \label{eq:framework-gamma}\\
n_{i,t,k} &\;=\; a_i \cdot \mathrm{Poisson}\!\bigl(\lambda_{i,k}\cdot \tilde q_{i,t}\bigr), \label{eq:framework-counts}
\end{align}
so a cascade first draws an engagement intensity $\lambda_{i,k}$ around its expected total $\mu_{i,k}$, which is the product of how widely the algorithm circulates the seed ($\mathit{E}_i$) and how likely an exposed user is to act ($p_{i,k}$). The cascade then accumulates the counts $n_{i,t,k}$ hour by hour in proportion to its exposure schedule, where $\tilde q_{i,t}=q_{i,t}/\mathit{E}_i$ is the share of its exposure that falls in hour $t$. Summing over hours, an active cascade's objective-$k$ total $N_{i,k} = \sum_t n_{i,t,k}$ therefore follows a Negative-Binomial distribution,
\begin{equation}
P\bigl(N_{i,k} = n \mid a_i = 1\bigr) \;=\; \frac{\Gamma(n+r)}{\Gamma(r)\,n!}\left(\frac{r}{r+\mu_{i,k}}\right)^{\!r}\left(\frac{\mu_{i,k}}{r+\mu_{i,k}}\right)^{\!n},
\label{eq:framework-nb}
\end{equation}
with mean $\mu_{i,k}$ and variance $\mu_{i,k} + \mu_{i,k}^{2}/r$. The dispersion $r$ is shared across the four objectives and governs how unevenly engagement concentrates across cascades, since the variance grows as $r$ falls.

The three calibrated parameters $(\pi_{\mathrm{active}}, \beta, r)$ are set so that the simulated root-reply counts (i.e., the number of replies received by the tweet that starts a conversation) match the share of tweets with no replies, the average reply count among tweets with at least one reply, and the variance of those counts after dropping the largest 5\%. It is worth emphasizing that the calibration is performed once under the additive rule and remains fixed for all subsequent analyses, so any difference in propagation between scoring rules reflects the difference in the score-aggregation rules. On the USC X corpus, this gives $\pi_{\mathrm{active}}=0.216$, $\beta=100.9$, and $r=0.318$, and we describe the complete, detailed simulation procedure in Appendix~\ref{app:algorithm}.

% ==============================================================
\section{Empirical Analysis and Results}
\label{sec:methods}
\label{sec:results}
% ==============================================================

We now initialize our simulation framework on the USC X corpus, where the user pool and the seed population are sampled from observed accounts and tweets, the ranker is trained on observed engagement, and the event model is calibrated to the observed count distribution. We proceed with our empirical analysis in two stages: in Stage~1, we perform a \emph{label-free validation} and ask whether the simulator reproduces real cascade patterns, and in Stage~2, we perform the \emph{hypothesis test} and ask whether changing the aggregation form differentially affects the propagation of low- and high-credibility tweets. The hypothesis test compares five scoring rules, summarized in Table~\ref{tab:regimes}.

\begin{table}[h]
\centering
\small
\caption{Scoring rules compared in the analysis.}
\label{tab:regimes}
\begin{tabular}{>{\raggedright\arraybackslash}p{0.30\linewidth} >{\raggedright\arraybackslash}p{0.36\linewidth} >{\raggedright\arraybackslash}p{0.26\linewidth}}
\toprule
Scoring rule & Formula & Role \\
\midrule
Additive (baseline) & $S=\sum_k w_k\,p_k$ & Production-style baseline \\
Slow-gates-fast multiplicative (F1) & $S=S_{\mathrm{slow}}\,(1+\alpha\,S_{\mathrm{fast}})$ & Reflective-gated form \\
Ratio-correction (F2) & $S=S_{\mathrm{additive}}\,/\,\bigl(1+\alpha\,S_{\mathrm{fast}}/(S_{\mathrm{slow}}+\epsilon)\bigr)$ & Reflective-gated form \\
Reflective-threshold sigmoid (F3) & $S=S_{\mathrm{additive}}\cdot\sigma\bigl((S_{\mathrm{slow}}-\theta)/\mathrm{scale}\bigr)$ & Reflective-gated form \\
Retuned additive (control) & Additive with $w_{\mathrm{reply}}$ halved & Parameter-only control \\
\bottomrule
\end{tabular}
\end{table}

\subsection{Stage 1: Validation of Our Simulation Framework}
\label{sec:methods-stage1}
\label{sec:results-stage1}

We validate the simulator on four metrics, each computed identically on the observed corpus and on the simulator output. These are the distribution of root reply counts, the distribution of aggregate engagement counts, the per-tweet reflective share of engagement (replies over the total engagement counts), and the distribution of time to peak (the hour carrying the maximum reply rate). We compare the observed and simulated means of each metric with two-sample Welch $t$-tests, and report the comparison in Table~\ref{tab:stage1}. We find that none of the four differences is statistically significant (all $p > 0.10$), and the simulated and observed curves track closely through the body of each distribution (Figure~\ref{fig:stage1}). We additionally report the Kolmogorov--Smirnov statistics for the reflective share in Appendix~\ref{app:calib}. The four metrics, therefore, validate the simulator at the level of engagement counts.

\begin{table}[h]
\centering
\caption{Stage 1 validation statistics.}
\label{tab:stage1}
\begin{tabular}{lcccc}
\toprule
Metric & Observed mean & Simulated mean & Welch $t$ & $p$ \\
\midrule
Root reply count            & $8.20$  & $8.88$  & $-0.43$ & $0.67$ \\
Aggregate engagement        & $17.89$ & $15.98$ & $0.85$  & $0.39$ \\
Reflective engagement share & $0.097$ & $0.099$ & $-0.52$ & $0.60$ \\
Time to peak (hours)        & $1.66$  & $1.51$  & $1.53$  & $0.13$ \\
\bottomrule
\end{tabular}
\end{table}

\begin{figure}[h]
\centering
\includegraphics[width=0.7\linewidth]{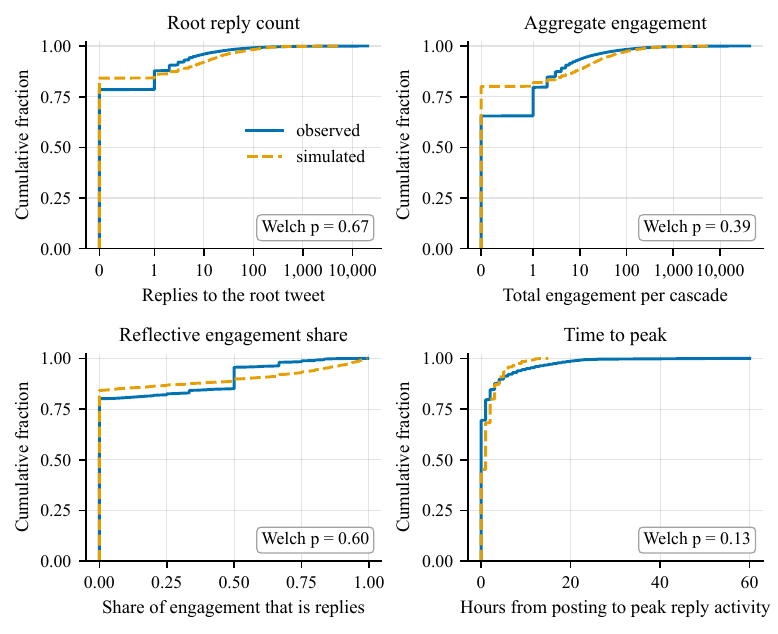}
\caption{Stage~1 validation. Simulated versus observed distributions for the four validation metrics.}
\label{fig:stage1}
\end{figure}

\begin{table}[h]
\centering
\small
\caption{Impressions received by tweets on the platform and exposure allocated by the simulator.}
\label{tab:expvalid}
\begin{tabular}{lcc}
\toprule
Statistic & Observed impressions & Simulated exposure \\
\midrule
Low-credibility share of the total & $20.6\%$ & $22.1\%$ \\
Credibility gap in shares (low $-$ high) & $-58.8$ pts $[-70.5, -43.4]$ & $-55.9$ pts $[-59.4, -52.4]$ \\
Mean ratio, low to high & $0.334$ $[0.223, 0.508]$ & $0.372$ $[0.329, 0.415]$ \\
Gini across tweets & $0.972$ & $0.906$ \\
Top 10\% share of the total & $98.1\%$ & $82.5\%$ \\
$p90 / p50$ & $19.8$ & $14.3$ \\
\bottomrule
\end{tabular}
\par\vspace{2pt}{\SingleSpacedXI\footnotesize \textit{Note.} Computed on the $18{,}884$ labeled original tweets. Simulated exposure is the additive-rule allocation for those same tweets. Brackets give $95\%$ bootstrap intervals over $2{,}000$ resamples stratified by credibility class.\par}
\end{table}

In addition, the impression record at the platform enables us to validate the exposure layer. To do so, we compute the exposure that the simulator allocates to each of the 18,884 labeled original tweets of the held-out slice under the additive rule, and we compare it with the number of impressions that the same tweets received on the platform. As reported in Table~\ref{tab:expvalid}, we find that the simulator reproduces the observed allocation closely. The credibility gap in exposure share is $-55.9$ points against an observed $-58.8$, and the mean exposure ratio between the classes is $0.372$ against an observed $0.334$, with the bootstrap intervals overlapping in both cases. The concentration of exposure across tweets is close as well, with a Gini of $0.906$ against $0.972$ and $82.5\%$ of simulated exposure falling on the top decile of tweets against $98.1\%$ of impressions. With the simulator validated at both layers, we now turn to the hypothesis tests.

\subsection{Testing Hypothesis 1: The Aggregation Form Narrows the Gap}
\label{sec:methods-stage2}
\label{sec:results-stage2}

To test Hypothesis~\ref{prop:closure}, which predicts that reflective-gated aggregation would narrow the propagation gap relative to the additive baseline, we draw 2{,}500 low-credibility and 2{,}500 high-credibility seed tweets, simulate each scoring rule 100 times, and measure two different outcomes: the \emph{exposure} as the direct measure of algorithm output, and the \emph{cascade size} as its downstream consequence once simulated users respond.

As reported in Table~\ref{tab:arch}, we find that reflective-gated aggregation narrows the exposure gap significantly under all three reflective-gated forms, with the reflective-threshold sigmoid producing the largest reduction. Meanwhile, the parameter-only control moves the gap significantly in the \emph{opposite} direction. Since the reflective-gated forms keep every weight of the additive baseline and change only the functional form, the opposite movements of the gap are attributed to the aggregation form, which supports Hypothesis~\ref{prop:closure}. The same pattern holds in every variation of the robustness checks in Section~\ref{sec:robust}, so the gap reduction does not depend on any single configuration choice.

In addition, as reported in Table~\ref{tab:arch}, we find that the cascade size contrast of the reflective-gated aggregation is also significantly negative, indicating that it narrows the credibility gap in downstream cascade size as well. The parameter-only control, meanwhile, widens it slightly. The downstream effect also becomes stronger as the ranker improves, since when we retrain the ranker on two and then five times as much data while holding the evaluation corpus fixed (Check~7 in Appendix~\ref{app:robust}), we find that the cascade size contrast grows in magnitude, while the parameter-only control still moves the gap in the opposite direction.

\begin{table}[h]
\centering
\caption{Exposure and cascade size contrasts against the additive baseline.}
\label{tab:arch}
\begin{tabular}{lcc}
\toprule
Scoring rule & Exposure contrast & Cascade size contrast \\
\midrule
Slow-gates-fast multiplicative (F1) & $-16.14$$^{***}$ (4.34) & $-3.79$$^{***}$ (0.94) \\
Ratio-correction (F2)               & $-2.63$$^{***}$ (0.28)  & $-0.76$$^{**}$ (0.38) \\
Reflective-threshold sigmoid (F3)   & $-37.83$$^{***}$ (5.05) & $-8.64$$^{***}$ (2.48) \\
Retuned additive (control)          & $+1.52$$^{***}$ (0.11)  & $+0.54$$^{**}$ (0.26) \\
\bottomrule
\end{tabular}
\par\vspace{2pt}{\SingleSpacedXI\footnotesize \textit{Note.} Bootstrap standard errors in parentheses. $^{*}\,p<0.10$, $^{**}\,p<0.05$, $^{***}\,p<0.01$.\par}
\end{table}

\subsection{Testing Hypothesis 2: Re-Tuning the Weights Does Not Help}
\label{sec:sign-asymmetry}

Hypothesis~\ref{prop:parameter} predicts that adjusting the weights of the additive rule, while keeping its additive form, would not narrow the propagation gap. The first piece of evidence is already visible in Table~\ref{tab:arch}, since the parameter-only control moves it significantly in the opposite direction. Since our control changes only a single weight, we next ask whether a more systematic weight configuration would do better. We therefore compute the gap reduction across 2{,}083 distinct weight configurations that together cover the range of ways of dividing the weights among the engagement objectives, and compare them against the reduction achieved by the reflective threshold. Surprisingly, we find that re-weighting usually backfires, as $74.3\%$ of the configurations \emph{widen} the credibility gap, while only $1.0\%$ reach even half of the reflective threshold's gap reduction, and none of them manage to reach $80\%$. Even the most obvious fix, setting the retweet and like weights to zero so that only replies and quote-tweets count, achieves only $6.9\%$ of the threshold form's gap reduction. Putting all the weight on retweets and likes does the opposite, widening the gap by more than a third of the reflective threshold's gap reduction, consistent with what our mechanism predicts.

A platform, however, cannot choose weights freely, since the weights also determine what the feed shows, and thus how well the feed serves the business objectives of the platform. We therefore trace the best gap reduction attainable at each level of fidelity (i.e., Spearman correlation) to the published content ordering in Table~\ref{tab:simplex}, and we find that reducing the gap and preserving the published ordering are in direct conflict. The unconstrained best configuration reaches $74.6\%$ of the gap reduction compared to the reflective threshold, but it does so by placing all weight on the single quote objective, which amounts to a substantially different ranking model. Once the ordering is held close to the deployed one, the attainable reduction collapses to under a quarter. Therefore, re-weighting does not provide a usable alternative, since the configurations that would help most are the ones a platform cannot adopt, which supports our Hypothesis~\ref{prop:parameter}.

\begin{table}[h]
\centering
\caption{Best attainable gap reduction by re-weighting, at each level of fidelity to the published ordering.}
\label{tab:simplex}
\begin{tabular}{lcc}
\toprule
Fidelity to the published ordering & Configurations & Best gap reduction attained \\
\midrule
Ordering closely preserved ($\rho \geq 0.95$) & 1{,}489 & $24.3\%$ \\
Ordering loosely preserved ($\rho \geq 0.90$) & 1{,}946 & $64.0\%$ \\
Unconstrained & 2{,}083 & $74.6\%$ \\
\bottomrule
\end{tabular}
\par\vspace{2pt}{\SingleSpacedXI\footnotesize \textit{Note.} Gap reduction is expressed as a share of the reduction the reflective threshold achieves. $\rho$ is the Spearman correlation between the configuration's content ordering and the published one.\par}
\end{table}

It is worth noting that X has since re-weighted its own algorithm in production, as its 2026 open-sourcing of the current algorithm \citep{xai2026algorithm} keeps the additive form and changes only the weights, cutting the reply weight from $13.5$ to $5.0$ and setting the quote weight equal to it. Repeating our previous analysis on this newly deployed weight configuration, we find that it would leave the credibility gap statistically unchanged (gap $+0.8$, s.e.\ $0.71$) while barely changing the content ordering ($\rho = 0.995$), consistent with our findings.

\subsection{Testing Hypothesis 3: The Gate Acts on the Engagement Mix}
\label{sec:counts}

Hypothesis~\ref{prop:targeting} predicts that a reflective-gated rule moves a tweet down according to its engagement mix and not according to its credibility. By the \emph{engagement mix} of a tweet, we mean how much of its engagement comes from instant reactions (retweets and likes) and how much from thoughtful responses (replies and quotes). This hypothesis matters because the gate never observes whether a tweet links to a low- or high-credibility source, so it can reduce the exposure of low-credibility content only if that content differs from high-credibility content in what the gate does observe. We test the hypothesis in four steps. We first check that the division of engagement into a fast and a slow class matches how engagement behaves in the corpus, then show that low-credibility content draws more of its engagement from instant reactions, both in the observed records and in the predictions of the simulated ranker. We then examine how the gate treats the two classes in the simulation, and finally rule out two other explanations for why the gate reduces the exposure of low-credibility content.

\subsubsection{Grounding the Cognitive Partition}
\label{sec:partition}

We first check that the division into a fast and a slow class is not arbitrary. For each account and content feature, we take the tweets in the top fifth and in the bottom fifth of that feature and compare how reply-heavy their engagement is. Table~\ref{tab:features} reports the difference between the two fifths for each feature. We find that features which indicate standing and substance, such as the number of lists an author appears on and the length of the tweet, go with reply-heavy engagement, while a low-credibility link, a high share of capital letters, and angry or disgusted language go with retweet-heavy engagement. This is what dual-process theory predicts, since replying requires reading a tweet and writing a response while retweeting takes a single click, and it supports treating replies and quote-tweets as the slow class and retweets and likes as the fast class. To make sure that our results do not depend on this particular division, we repeat the analysis under three other divisions in Appendix~\ref{app:robust}, and we find that the gate narrows the credibility gap under all four divisions.

\begin{table}[h]
\centering
\caption{Engagement mix by account and content feature.}
\label{tab:features}
\begin{tabular}{lll}
\toprule
Feature & Reflective-to-reactive gap & Interpretation \\
\midrule
Listed count (log)      & $+0.286$ & prestige $\rightarrow$ reply-heavy \\
Followers (log)         & $+0.258$ & prestige $\rightarrow$ reply-heavy \\
Text length             & $+0.219$ & substance $\rightarrow$ reply-heavy \\
Paid verification       & $+0.214$ & verification $\rightarrow$ reply-heavy \\
Follower-to-friend ratio (log) & $+0.191$ & influence $\rightarrow$ reply-heavy \\
Mention count           & $+0.134$ & conversation $\rightarrow$ reply-heavy \\
Low-credibility URL     & $-0.118$ & low-credibility $\rightarrow$ retweet-heavy \\
Capitalization ratio    & $-0.113$ & all-caps $\rightarrow$ retweet-heavy \\
Anger/disgust emotion   & $-0.027$ & emotional $\rightarrow$ slightly retweet-heavy \\
\bottomrule
\end{tabular}
\par\vspace{2pt}{\SingleSpacedXI\footnotesize \textit{Note.} The gap compares how reply-heavy the engagement of tweets is in the top and bottom fifth of each feature, and a positive value means that the top fifth is more reply-heavy.\par}
\end{table}

\subsubsection{Observed Engagement by Credibility Class}
\label{sec:counts-observed}

We next ask whether low-credibility content does draw more of its engagement from instant reactions. Since high-credibility content draws more engagement in total, a direct comparison would mix the amount of engagement with its composition, so we compare low- and high-credibility tweets that received the same total engagement. Table~\ref{tab:observed} reports this comparison, which uses the labeled tweets directly and involves no simulation. We find that, at the same total engagement, low-credibility tweets receive significantly more retweets and fewer replies than high-credibility tweets, while their likes and their total number of replies, retweets, and quotes are about the same. Replies make up $33.0\%$ of the replies, retweets, and quotes that a low-credibility tweet draws on average, against $46.4\%$ for a high-credibility tweet. Low-credibility content therefore draws more of its engagement from instant reactions, which is the difference that Hypothesis~\ref{prop:targeting} requires.

\begin{table}[h]
\centering
\caption{Observed per-tweet engagement by credibility class, for tweets matched on total engagement.}
\label{tab:observed}
\begin{tabular}{lcc}
\toprule
Metric (mean per tweet) & low-credibility ($n=3{,}116$) & high-credibility ($n=6{,}853$) \\
\midrule
Replies                       & $1.6$  & $2.1$ \\
Retweets                      & $3.7$  & $3.2$ \\
Likes                         & $8.1$  & $8.0$ \\
Reply share of engagement     & $33.0\%$ & $46.4\%$ \\
Cascade-size proxy            & $5.5$ & $5.6$ \\
\bottomrule
\end{tabular}
\par\vspace{2pt}{\SingleSpacedXI\footnotesize \textit{Note.} Reply share is replies over replies, retweets, and quotes, averaged over tweets with at least one such action. The cascade-size proxy is the sum of replies, retweets, and quotes.\par}
\end{table}

\subsubsection{Predicted Engagement by Credibility Class}
\label{sec:predicted}

The gate acts on the predictions of the ranker at the moment a tweet is scored, and not on the engagement that the tweet later receives, so the same difference must appear in the predictions. Table~\ref{tab:predicted} reports the predictions of the ranker for the same matched tweets. We find that the ranker predicts higher probabilities of a retweet and a like for low-credibility tweets but lower probabilities of a reply and a quote, so that the predicted engagement of a low-credibility tweet leans more toward instant reactions. The ratio of predicted fast- to slow-class engagement is $2.97$ for low-credibility tweets against $2.62$ for high-credibility tweets (Figure~\ref{fig:mechanism}a). Since the gate compares the predicted slow-class engagement of a tweet with a threshold, a low-credibility tweet is less likely to clear it.

\begin{table}[h]
\centering
\caption{Ranker-predicted engagement by credibility class, for the matched tweets.}
\label{tab:predicted}
\begin{tabular}{lcc}
\toprule
Quantity (mean per tweet) & low-credibility & high-credibility \\
\midrule
$P(\text{reply})$    & $0.238$ & $0.282$ \\
$P(\text{retweet})$  & $0.323$ & $0.303$ \\
$P(\text{like})$     & $0.485$ & $0.466$ \\
$P(\text{quote})$ & $0.103$ & $0.116$ \\
\textbf{Predicted reactive/reflective} & $\mathbf{2.973}$ & $2.620$ \\
$S_{\mathrm{additive}}$ & $3.935$ & $4.417$ \\
$S_{\mathrm{F1}}$  & $6.674$ & $7.283$ \\
\bottomrule
\end{tabular}
\par\vspace{2pt}{\SingleSpacedXI\footnotesize \textit{Note.} Computed on the matched tweets of Table~\ref{tab:observed}.\par}
\end{table}

\begin{figure}[t]
\centering
\includegraphics[width=0.7\linewidth]{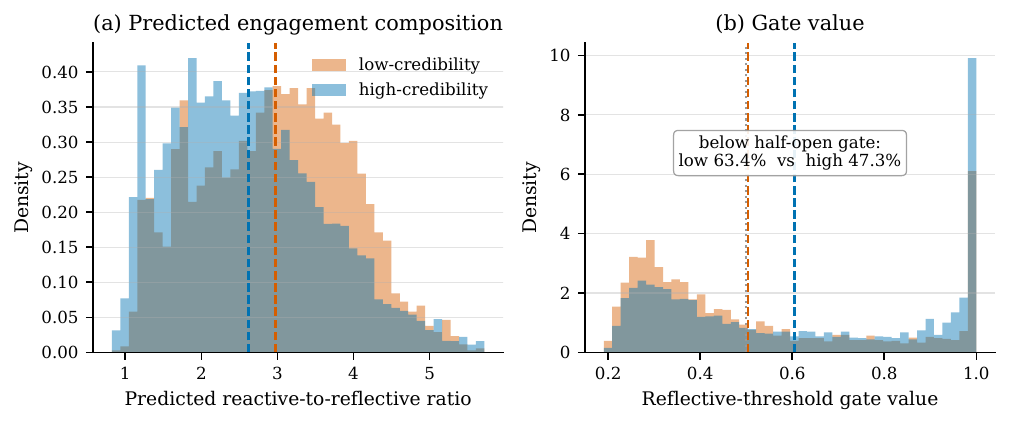}
\caption{Class-level distributions. (a)~The ranker's predicted reactive-to-reflective ratio over the matched tweets of Table~\ref{tab:observed}. (b)~The reflective-threshold gate value over the hypothesis-test seeds.}
\label{fig:mechanism}
\end{figure}

\subsubsection{How the Gate Acts}
\label{sec:mechanism}

We then look at what the gate does with these predictions in the simulation. The reflective threshold multiplies the additive score of a tweet by a gate value between $0$ and $1$, and the gate value falls below $0.5$ once the predicted slow-class engagement of the tweet drops below the threshold. We find that $63.4\%$ of low-credibility tweets receive a gate value below $0.5$, against $47.3\%$ of high-credibility tweets (Figure~\ref{fig:mechanism}b). Since the gate never observes the credibility label, this difference can only come from the different predicted engagement of the two classes. It is worth noting that the multiplicative form and the threshold form reach their effect in different ways. Under the multiplicative form, the share of tweets that are moved down, that is, given a lower position than under the additive rule, is about the same in both classes ($43.1\%$ of low-credibility against $44.2\%$ of high-credibility tweets), while under the reflective threshold $52.8\%$ of low-credibility tweets are moved down against $38.6\%$ of high-credibility tweets, which is why it produces the largest effect of the three forms in Table~\ref{tab:arch}.

\subsubsection{Ruling Out Alternative Explanations}
\label{sec:falsification}

We consider two alternative explanations for why the gate reduces the exposure of low-credibility content, and we test each in turn.

\emph{Alternative 1: high-credibility content has more predicted engagement of every kind.} If this were the explanation, a gate on the fast class would reduce the exposure of low-credibility content in the same way as our gate on the slow class. To test it, we build a placebo gate that has the same threshold form as our gate but acts on the predicted fast-class engagement, with its threshold and steepness set at the matching percentiles of the fast-class predictions. We find that the two gates move down opposite types of tweets. Under our gate, a tweet is moved down further if the share of its score that comes from instant reactions is larger, while under the placebo gate, such tweets are moved up, and the placebo gate narrows the exposure gap by $12.1$ (s.e.\ $0.82$) less than our gate does. The larger amount of predicted engagement of high-credibility content, therefore, does not explain which tweets our gate moves down.

\emph{Alternative 2: the gate acts on the amount of predicted slow-class engagement and not on the mix.} If this were the explanation, how far a tweet is moved down would follow its total predicted slow-class engagement and not its engagement mix. We find the opposite, as the demotion of a tweet increases with its predicted reactive-to-reflective ratio, while it is barely related to its total predicted slow-class engagement, and the relationship with the ratio holds among tweets with the same predicted slow-class engagement and within each credibility class separately. The mix can matter even though the gate never observes it, because the gate scales the entire additive score, so of two tweets with the same predicted slow-class engagement, the one whose score relies more on instant reactions is ranked higher and has further to fall.

Together, the two tests show that the gate reduces the exposure of low-credibility content because of the composition of its predicted engagement, which supports Hypothesis~\ref{prop:targeting}.

\subsection{Robustness Checks}
\label{sec:methods-robust}
\label{sec:robust}

In this section, we check that our findings do not depend on a specific configuration, by varying the configuration in eight ways and varying the hyperparameters of the gate, which gives the nine checks summarized in Table~\ref{tab:robustsummary}, with their details reported in Appendix~\ref{app:robust}. We find that the observed pattern holds in all 46 variations and at all 23 hyperparameter settings, with every exposure contrast significant in the predicted direction and the parameter-only control always moving in the opposite direction. It is worth noting that the reflective threshold is the strongest of the three reflective-gated forms, producing $2.3\times$ the exposure effect of the multiplicative gate and the largest cascade size contrast, and we carry it forward as our solution in Section~\ref{sec:solution}.

\begin{table}[h]
\centering
\small
\caption{Summary of robustness checks.}
\label{tab:robustsummary}
\begin{tabular}{c >{\raggedright\arraybackslash}p{0.36\linewidth} >{\raggedright\arraybackslash}p{0.44\linewidth}}
\toprule
Check & Variation & Finding \\
\midrule
1 & Aggregation form (F1--F3 and the retuned control) & All three reflective-gated forms narrow the gap significantly, the control moves it the opposite way \\
2 & Slow/fast partition (four assignments) & Significantly negative under all four \\
3 & Ranker training seed (five retrainings) & Contrasts cluster in $[-21.3, -15.5]$, every one significant \\
4 & Corpus period (adjacent week-long slice) & $-16.3$ on the adjacent slice against $-16.1$ at baseline \\
5 & Partition $\times$ training seed ($4\times5$ combinations) & Significantly negative in all 20 combinations \\
6 & User-pool size (25K, 50K, 100K users) & Identical contrasts at all three sizes \\
7 & Training-corpus size (2M and 5M tweets) & Strengthens monotonically at both layers \\
8 & Ranker design (four feature subsets, MLP, transformer) & Sign pattern holds in every variant \\
9 & Gate hyperparameters (23 grid points) & Significantly negative at every point, growing with $\alpha$ and $\theta$ \\
\bottomrule
\end{tabular}
\end{table}

To close the section, Table~\ref{tab:hypotheses} summarizes the verdicts on the three hypotheses. We find that all three are supported, and that Hypothesis~\ref{prop:closure} holds both for exposure and for cascade size.

\begin{table}[H]
\centering
\small
\caption{Summary of hypothesis-test results.}
\label{tab:hypotheses}
\begin{tabular}{>{\raggedright\arraybackslash}p{0.22\linewidth} >{\raggedright\arraybackslash}p{0.13\linewidth} >{\raggedright\arraybackslash}p{0.53\linewidth}}
\toprule
\textbf{Hypothesis} & \textbf{Sections} & \textbf{Result} \\
\midrule
H1: The aggregation form narrows the gap & \ref{sec:results-stage2}, \ref{sec:robust} & Supported at both layers, on exposure under all three reflective-gated forms and in all 46 robustness variations, and on cascade size under the baseline ranker. \\
\midrule
H2: Re-tuning the weights does not help & \ref{sec:results-stage2}, \ref{sec:sign-asymmetry} & Supported, as the retuned control moves the gap in the opposite direction and no weight configuration reaches the threshold form's gap reduction. \\
\midrule
H3: The gate acts on the engagement mix & \ref{sec:counts} & Supported, as demotion tracks the predicted reactive-to-reflective ratio, and the placebo gate and the correlation analysis show that neither alternative explanation accounts for the targeting. \\
\bottomrule
\end{tabular}
\end{table}

% ==============================================================
\section{Validating Our Proposed Solution}
\label{sec:solution}
% ==============================================================

\subsection{The Solution}
\label{sec:operating}

We now evaluate the effectiveness of the reflective-threshold gate as a solution. Under this gate, the recommendation score of a tweet is scaled down when the ranker predicts that the tweet will draw little thoughtful engagement, such as replies and quotes, and the threshold $\theta$ sets how much thoughtful engagement a tweet needs to avoid this reduction, so that a higher $\theta$ is more demanding. To measure the benefit of the gate, we compare the share of total exposure that low-credibility tweets receive under the gate with the share they receive under the additive rule. As shown in Figure~\ref{fig:operating}, we find that the benefit grows steadily as $\theta$ increases, with the exposure share of low-credibility tweets falling by about $8\%$ once $\theta$ reaches $2$. It is worth emphasizing that high-credibility tweets do not lose exposure at any setting, since they are more likely to pass the threshold, and receive the exposure that low-credibility tweets lose. Therefore, a platform would not need to sacrifice the exposure of mainstream content in order to reduce the exposure of low-credibility content.

A platform would also want to know whether the gate reduces the engagement that its feed generates, and we find that it does not for the following two reasons. First, a feed shows its highest-scored tweets first, and we find that the top 10\% of tweets under the gate are the same tweets as under the additive rule, while at least 98\% of the top 20\% are the same. The gate therefore mainly reorders tweets in the middle of the ranking, where the tweets that rely on instant reactions are located, and leaves the top of the feed unchanged. Second, the tweets that gain exposure under the gate are the tweets that the ranker expects to draw more engagement, so that the same amount of exposure would produce more engagement. Specifically, when the total exposure is held fixed and allocated in proportion to each tweet's score, the expected engagement per exposure rises by about 17\% at $\theta = 1$. Table~\ref{tab:engcost} reports these measures for each value of $\theta$.

\begin{table}[h]
\centering
\caption{Benefit and engagement cost of the reflective-threshold gate at different values of $\theta$.}
\label{tab:engcost}
\small
\begin{tabular}{ccccc}
\toprule
$\theta$ & \makecell{Low-credibility\\exposure share} & \makecell{High-credibility\\exposure share} & \makecell{Expected engagement\\per exposure} & \makecell{Top-20\%\\overlap} \\
\midrule
$0.5$ & $-2.9\%$ & $+2.3\%$ & $+9.4\%$  & $99.7\%$ \\
$1.0$ & $-5.0\%$ & $+3.9\%$ & $+17.4\%$ & $99.3\%$ \\
$1.5$ & $-7.0\%$ & $+5.5\%$ & $+27.2\%$ & $99.0\%$ \\
$2.0$ & $-8.1\%$ & $+6.4\%$ & $+37.4\%$ & $98.6\%$ \\
$2.2$ & $-8.2\%$ & $+6.5\%$ & $+41.4\%$ & $98.4\%$ \\
\bottomrule
\end{tabular}
\par\vspace{2pt}{\SingleSpacedXI\footnotesize \textit{Note.} Exposure shares and expected engagement per exposure are changes relative to the additive rule at the same total exposure.\par}
\end{table}

It is worth noting that about 97\% of the tweets in our corpus carry no credibility label, so the results above do not show how the gate affects them. To examine this, we mix 5{,}000 unlabeled tweets with 5{,}000 labeled ones and compare the position of each tweet in the ranking with and without the gate. We find that unlabeled tweets move slightly up on average, and although they make up half of the mixed set, they account for only 18\% of the 10\% of tweets that the gate moves down the most, while the unlabeled tweets that the gate does move down are those that rely more on instant reactions. In addition, since the benefit grows only slowly once $\theta$ exceeds $1.5$, any value of $\theta$ between 1 and 2 would achieve most of the reduction. Therefore, $\theta$ becomes the single number that a platform needs to set or optimize based on its own target.

\begin{figure}[t]
\centering
\includegraphics[width=0.5\linewidth]{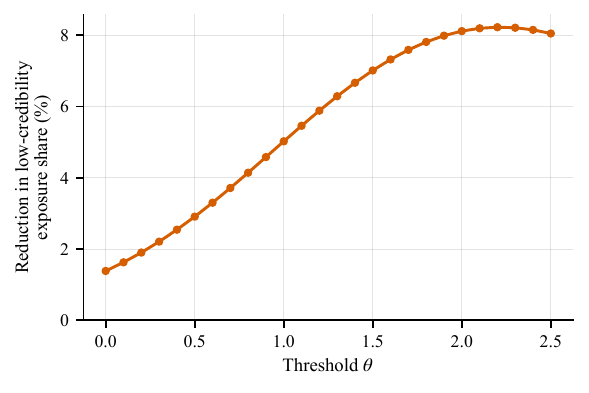}
\caption{Reduction in the exposure share of low-credibility tweets as the threshold $\theta$ rises.}
\label{fig:operating}
\end{figure}

\subsection{Evaluation and Validation}
\label{sec:validation}

A natural concern about our solution is that the result may be circular, since a gate that requires thoughtful engagement would, by design, hold back the content that draws the least of it, and low-credibility content is part of that content. To address this concern, we validate the gate with the following three checks. First, the \emph{targeting} check examines whether the gate moves each tweet down according to how much of its score comes from instant reactions, and not simply according to its credibility class. Second, the \emph{propagation} check examines whether this effect on the ranking also narrows the credibility gap in how far content spreads in the simulation. Third, the \emph{claim-level} check examines whether the gate also moves down tweets whose individual claims are rated as misleading.

\subsubsection{Targeting and Propagation}

To start with, we expect that the quantity the gate should act on is the share of a tweet's score that comes from fast signals, which we call its \emph{substitutability share}. As shown in Figure~\ref{fig:validation}a, we find that the higher the substitutability share of a tweet, that is, the more of its score comes from predicted likes and retweets, the further the gate moves it down. Specifically, low-credibility tweets have a higher substitutability share on average, which suggests that the additive rule ranks them highly partly because of instant reactions. As a result, low-credibility tweets make up $61\%$ of the quarter of tweets that the gate moves down the most, against $44\%$ of the quarter that it moves down the least.

In addition, the same relationship holds when we look at low-credibility and high-credibility tweets separately, which indicates that the gate does not treat the two classes differently, and moves low-credibility tweets down more only because more of them have a high substitutability share. The gate also barely moves tweets that are predicted to draw plenty of thoughtful engagement, in either class, which distinguishes our gate from a simple penalty on low-credibility sources. The propagation check leads to the same conclusion, as we find that the gate significantly narrows the gap in cascade size between low- and high-credibility content (Table~\ref{tab:arch}), so the effect on the ranking also carries over to how far content spreads in the simulation.

\subsubsection{The Claim-Level Limit}

For the claim-level check, we use 400 labeled tweets whose claims were rated by an advanced LLM as misleading, unclear, or not misleading, without access to the credibility of their linked sources. We find that the gate does not move the tweets rated misleading down more than the other tweets, as shown in Figure~\ref{fig:validation}b, since these tweets do not rely especially on instant reactions, and the gate in fact moves them down slightly less. Therefore, our gate reduces the advantage that low-credibility sources gain from instant reactions, while it has little connection to whether an individual claim is true in this corpus, so we do not present it as a detector of misinformation.

\begin{figure}[t]
\centering
\includegraphics[width=0.8\linewidth]{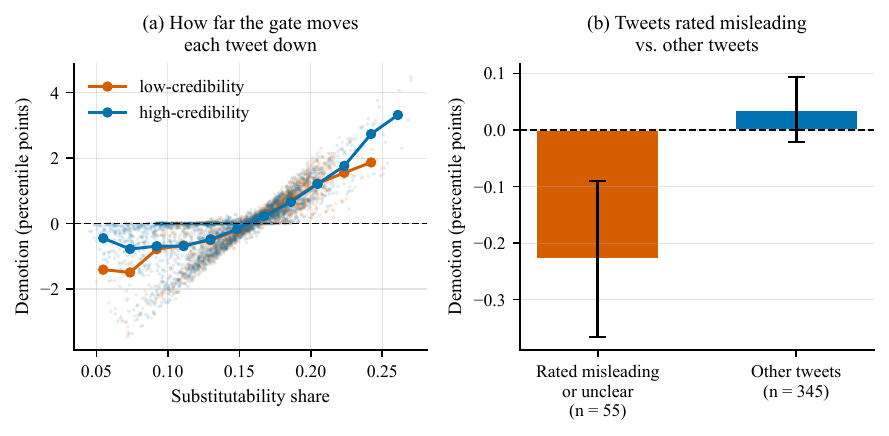}
\caption{Targeting validation. (a)~How far the reflective threshold moves each tweet down, against its substitutability share. (b)~How far it moves the tweets rated misleading and the other tweets down.}
\label{fig:validation}
\end{figure}

We next ask whether a stronger rule would pass the claim-level check, and we test two such rules that we fixed in advance. The first rule multiplies predicted thoughtful engagement and predicted instant reactions together, $S = S_{\mathrm{slow}}^{0.75}\,S_{\mathrm{fast}}^{0.25}$, so that a tweet needs both kinds of engagement to score highly, while the second rule multiplies the reflective-threshold score by a score built from the account and content features of Section~\ref{sec:partition}. We find that both rules appear to move misleading tweets down more, yet neither does so for the reason that our mechanism describes. Specifically, the first rule mostly moves down tweets that draw little engagement of any kind, since the product is small whenever either kind of engagement is small. The second rule works through the added feature score and not through how the engagement predictions are combined, and its effect is not statistically significant in this sample. The reason for this limit is that the engagement predictions carry little information about whether a claim is true. Therefore, while changing how the engagement predictions are combined is an effective way to address the amplification of instant reactions, passing the claim-level check would require giving the ranker additional inputs that carry information about the claims themselves.

To make sure that this limit does not come from using an LLM to rate the claims, we repeat the check with Community Notes, the crowd-sourced fact-checking notes that X makes public. Matching these notes against the tweets in our corpus gives 851 tweets with at least one note, 805 of which are rated as misleading, and these notes were written without any input from our models or labels. Compared with similar tweets without a note, matched on whether they are original posts and on how much engagement they draw, we find that the tweets with a note rely significantly less on instant reactions, so the reflective threshold would move them down slightly less than similar tweets. This result confirms the same limit with an independent set of labels that is twice as large as the claim-rated set, and we therefore test in the next subsection whether adding a content signal would overcome it.

\subsubsection{Extending the Gate with a Content Signal}
We subsequently ask whether the gate can use an additional input that does carry information about the claim, without losing the targeting we validated above. As a proof of concept, we train a simple text classifier on the 400 rated tweets to estimate how likely a tweet is to be misleading, and we always score each tweet with a version of the classifier that was not trained on it. We then multiply the reflective-threshold score by a content gate, $S = S_{\theta}\cdot g^{\eta}$, where $S_{\theta}$ is the reflective-threshold score, $g$ falls from one toward zero as the classifier considers a tweet more likely to be misleading, and the exponent $\eta$ sets how much weight the content gate receives. We find that at full weight ($\eta = 1$), the combined score passes the claim-level check but reverses the targeting result, since the tweets that the classifier flags tend to rely less on instant reactions, so that a content gate with full weight would fix the claim-level problem only by undoing the targeting.

However, since the reflective threshold acts on predicted engagement while the content gate acts on the text of the tweet, the two gates work largely independently of each other, so that giving the content gate a small weight would be able to keep both results (Figure~\ref{fig:frontier}). Over a range of $\eta$ values fixed in advance, we find that at $\eta = 0.15$, the tweets with the highest substitutability share are still moved down almost as much as without the content gate ($97\%$ of their original demotion), and the targeting retention in Figure~\ref{fig:frontier}b stays at $84\%$, while misleading tweets are now moved down significantly more than the other tweets. In addition, a variant that applies the content gate only to the $10\%$ of tweets that the classifier considers most likely to be misleading moves the tweets with the highest substitutability share down by the same amount as the reflective threshold alone. The content gate with a small weight therefore adds the ability to move misleading tweets down while keeping the targeting of the reflective threshold, so that the two gates complement each other. Overall, the reflective threshold remains our main solution, while the content gate shows that it can be combined with additional content signals, as we discuss in Section~\ref{sec:practice}.

Since the analysis above only examines how the combined score reorders tweets, we also run the combined score through the full cascade simulation of Section~\ref{sec:results-stage2}, using the variants fixed in advance in a single run, to see how far content spreads under it. We find that at $\eta = 0.15$, the combined score changes the exposure gap by $-38.2$ (s.e.\ $5.41$), which is not statistically different from the change under the reflective threshold alone, so adding the content gate does not weaken the main result. In addition, the gap in cascade size narrows significantly under every variant of the combined score, and the variant that applies the content gate only to the most suspicious tweets narrows the exposure gap even further ($-43.2$). It is worth noting that the content gate gives similar values to low- and high-credibility tweets, which confirms that the two gates act through separate channels, as the reflective threshold narrows the credibility gap while the content gate moves down misleading tweets without widening the gap again.

\begin{figure}[t]
\centering
\includegraphics[width=0.85\linewidth]{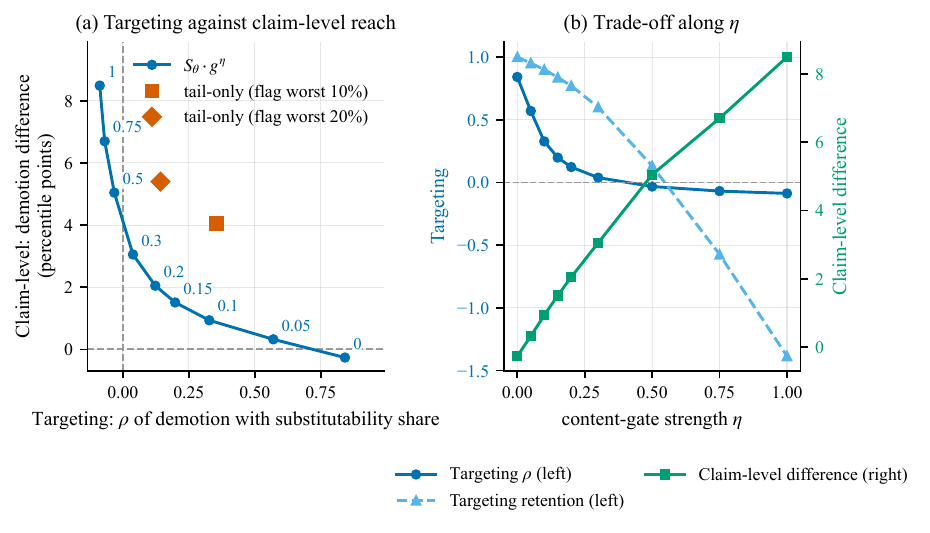}
\caption{The reflective threshold combined with the content gate. (a)~Targeting strength against how much further misleading tweets are moved down, for different weights $\eta$ of the content gate. (b)~Targeting and the claim-level difference as $\eta$ increases.}
\label{fig:frontier}
\end{figure}

% ==============================================================
\section{Discussion and Conclusion}
\label{sec:discussion}
% ==============================================================

\subsection{Implications for Researchers}
Our findings lead to the following implications for researchers in the academic community.

First, we locate the mechanism of differential misinformation propagation in a specific component, the score-aggregation layer, and identify the \emph{engagement fungibility} mechanism, which treats fast- and slow-class engagement as interchangeable evidence of quality. Previous studies at content-level \citep{vosoughi2018spread,brady2017emotion} and user-level \citep{pennycook2019lazy} explain why some content provokes reactive engagement, while we explain how a recommender turns this reactive engagement into differential exposure. As a result, the platform will be able to respond accordingly.

Second, our results add to the dual-process theory of recommender systems \citep{kleinberg2024challenge,agarwal2024system2}. These studies argue, using simple theoretical models, that a recommender that optimizes engagement favors the impulsive (System 1) responses of users over their reflective (System 2) responses, and they treat this as an intrinsic bias of engagement-based ranking. In contrast, we find that the bias comes from the way the engagement predictions are combined into a score. Every combination rule that requires predicted thoughtful engagement before predicted instant reactions can raise the score of a tweet narrows the credibility gap (Hypothesis~\ref{prop:closure}). Therefore, the bias would be located in one layer of the recommender that a platform can replace, and it would not require giving up engagement-based ranking.

Third, our study answers the question of what would happen to the same content if one part of the recommender were changed, by reconstructing a deployed recommender and changing one component inside a simulation that is calibrated to observed data. In addition, the way we estimate the dispersion of engagement counts would be useful for any simulation of engagement that has to fit a distribution to counts that are mostly zero and occasionally extremely large.

Finally, our results suggest that empirical studies of misinformation should examine the composition of engagement and not only its amount. In our corpus, high-credibility content receives more engagement in total, yet low-credibility content receives a larger share of its engagement from instant reactions, both in the observed counts and in the predictions of the ranker, and it is this composition that the gate responds to. Therefore, a study that compares only the total engagement would conclude that high-credibility content has the advantage, and would miss the difference that determines how the algorithm treats them.

\subsection{Implications for Practitioners}
\label{sec:practice}

For platforms, our findings suggest that the differential propagation of reactive content comes from a design choice in the recommender, which a platform would be able to change. Specifically, the property that produces the gap lies in the layer that combines the engagement predictions into a score, which a platform can change without retraining its models or changing its published weights.

A second implication is that the intervention spares the platform from having to rule on what is true. The gate acts on a predicted engagement pattern and not on a veracity label, so a platform need not decide which content is false, but only which engagement pattern it is unwilling to treat as evidence of quality. This is important in practice, since fact-checking at scale is costly and too slow to inform ranking. Platforms that already run content classifiers over posts, as X does with its safety and policy classifiers \citep{xai2026algorithm}, need not choose between the two approaches, because the engagement gate and a content signal work through separate channels and can be combined, while a small weight on the content signal limits how much an imperfect classifier can change what users see.

Deployment is also cheap, since the threshold is applied after scoring, to quantities the ranker already produces, so that it adds little to the cost of serving a feed. The gate would also be harder to manipulate than the additive rule, since it relies on predicted thoughtful engagement, such as replies and composed quote-tweets, which is far more costly for an adversary to produce at scale than one-click reactions, so that content would need engagement that takes effort to be amplified.

A final implication is associated with the incentive of a platform to adopt the gate. The economics of content moderation suggests that a platform also bears part of the cost of low-accuracy content, through advertiser pressure, regulation, and the loss of user trust, so that its best policy usually gives up some engagement in exchange for accuracy \citep{papanastasiou2020fake,candogan2020optimal}. Our findings suggest that this exchange would be more favorable than a platform might expect, since Figure~\ref{fig:operating} shows how much exposure the gate takes away from low-credibility content at each value of the threshold, while our simulation finds no loss of engagement at any threshold in Section~\ref{sec:operating}, as the gate leaves the top of the feed unchanged and moves exposure toward the tweets that are expected to draw more engagement. Therefore, a platform would be able to adopt the gate without sacrificing its engagement objective, which it would be able to verify on its own traffic.

\subsection{Limitations and Future Research}
\label{sec:limits}
\label{sec:conclusion}

Before turning to future research, we note the boundaries within which our claims hold. The first concerns the recommender we study, which is reconstructed and not replicated, since our four-objective ranker trains on far less data than the production system. Although the direction of the effect stays the same across ranker designs, and its size grows as the ranker is trained on more data (Checks~7 and~8), a production-grade ranker could behave differently, and since the simulator keeps the trained ranker fixed without retraining, our effect sizes describe the immediate change and not a new long-run equilibrium. The second concerns measurement, as exposure records what the algorithm distributes and not what people ultimately see. The third concerns the data, as our labels measure the credibility of the source and not the truth of each claim, so we show that the gate targets the reactive engagement pattern of low-credibility sources, while the content gate reaches misleading claims only as a proof of concept, as explained in Section~\ref{sec:validation}. In addition, the corpus is a keyword-filtered sample of a single platform over a single period.

These limitations suggest several directions for future research. First, replicating the content gate with sufficient statistical power would require claim-level labels for a few thousand tweets and a production-grade credibility model, together with the richer engagement objectives of the production system and Community Notes as inputs. Second, a field test on a cooperating platform would be able to measure the size of the effect in real-world applications. In addition, cross-platform studies would be able to test the cognitive-effort partition for the engagement actions of TikTok, YouTube, and Facebook, while an audit of past changes to score aggregation would be able to test the mechanism with observational data.

% ==============================================================
\bibliographystyle{informs2014}
\bibliography{references}

@inproceedings{agarwal2024system2,
  author    = {Agarwal, Arpit and Usunier, Nicolas and Lazaric, Alessandro and Nickel, Maximilian},
  title     = {System-2 Recommenders: Disentangling Utility and Engagement in Recommendation Systems via Temporal Point-Processes},
  booktitle = {Proceedings of the 2024 ACM Conference on Fairness, Accountability, and Transparency (FAccT '24)},
  year      = {2024},
  pages     = {1763--1773},
  address   = {Rio de Janeiro, Brazil},
  publisher = {ACM}
}

@inproceedings{balasubramanian2024usc,
  author    = {Balasubramanian, Ashwin and Zou, Vito and Narayana, Hitesh and You, Christina and Luceri, Luca and Ferrara, Emilio},
  title     = {A Public Dataset Tracking Social Media Discourse about the 2024 {U.S.} Presidential Election on {Twitter/X}},
  booktitle = {Workshop Proceedings of the 19th International AAAI Conference on Web and Social Media (CySoc 2025)},
  year      = {2025},
  note      = {Dataset: \url{https://github.com/sinking8/usc-x-24-us-election}}
}

@article{bandy2021curating,
  author  = {Bandy, Jack and Diakopoulos, Nicholas},
  title   = {Curating Quality? {H}ow {Twitter}'s Timeline Algorithm Treats Different Types of News},
  journal = {Social Media + Society},
  volume  = {7},
  number  = {3},
  year    = {2021},
  pages   = {20563051211041648}
}

@article{brady2017emotion,
  author  = {Brady, William J. and Wills, Julian A. and Jost, John T. and Tucker, Joshua A. and Van Bavel, Jay J.},
  title   = {Emotion shapes the diffusion of moralized content in social networks},
  journal = {Proceedings of the National Academy of Sciences},
  volume  = {114},
  number  = {28},
  pages   = {7313--7318},
  year    = {2017}
}

@article{corsi2023evaluating,
  author  = {Corsi, Giulio},
  title   = {Evaluating {Twitter}'s Algorithmic Amplification of Low-Credibility Content: An Observational Study},
  journal = {EPJ Data Science},
  volume  = {13},
  pages   = {18},
  year    = {2024}
}

@article{evans2013dual,
  author  = {Evans, Jonathan St.~B.~T. and Stanovich, Keith E.},
  title   = {Dual-process theories of higher cognition: Advancing the debate},
  journal = {Perspectives on Psychological Science},
  volume  = {8},
  number  = {3},
  pages   = {223--241},
  year    = {2013}
}

@misc{funke2026iffy,
  author       = {Golding, Barrett},
  title        = {{Iffy+} Mis/Disinfo Sites},
  howpublished = {\url{https://iffy.news/iffy-plus/}},
  year         = {2026},
  note         = {Accessed September 2026}
}

@article{huszar2022algorithmic,
  author  = {Husz{\'a}r, Ferenc and Ktena, Sofia Ira and O'Brien, Conor and Belli, Luca and Schlaikjer, Andrew and Hardt, Moritz},
  title   = {Algorithmic amplification of politics on {Twitter}},
  journal = {Proceedings of the National Academy of Sciences},
  volume  = {119},
  number  = {1},
  pages   = {e2025334119},
  year    = {2022}
}

@book{kahneman2011thinking,
  author    = {Kahneman, Daniel},
  title     = {Thinking, Fast and Slow},
  publisher = {Farrar, Straus and Giroux},
  address   = {New York},
  year      = {2011}
}

@article{kleinberg2024challenge,
  author  = {Kleinberg, Jon and Mullainathan, Sendhil and Raghavan, Manish},
  title   = {The Challenge of Understanding What Users Want: Inconsistent Preferences and Engagement Optimization},
  journal = {Management Science},
  volume  = {70},
  number  = {9},
  pages   = {6336--6355},
  year    = {2024}
}

@article{lin2023high,
  author  = {Lin, Hause and Lasser, Jana and Lewandowsky, Stephan and Cole, Rocky and Gully, Andrew and Rand, David G. and Pennycook, Gordon},
  title   = {High level of correspondence across different news domain quality rating sets},
  journal = {PNAS Nexus},
  volume  = {2},
  number  = {9},
  pages   = {pgad286},
  year    = {2023}
}

@misc{mbfc2026methodology,
  author       = {{Media Bias/Fact Check}},
  title        = {Methodology},
  howpublished = {\url{https://mediabiasfactcheck.com/methodology/}},
  year         = {2026},
  note         = {Accessed September 2026}
}

@article{pennycook2019lazy,
  author  = {Pennycook, Gordon and Rand, David G.},
  title   = {Lazy, not biased: Susceptibility to partisan fake news is better explained by lack of reasoning than by motivated reasoning},
  journal = {Cognition},
  volume  = {188},
  pages   = {39--50},
  year    = {2019}
}

@article{pennycook2021psychology,
  author  = {Pennycook, Gordon and Rand, David G.},
  title   = {The Psychology of Fake News},
  journal = {Trends in Cognitive Sciences},
  volume  = {25},
  number  = {5},
  pages   = {388--402},
  year    = {2021}
}

@article{stanovich2000individual,
  author  = {Stanovich, Keith E. and West, Richard F.},
  title   = {Individual differences in reasoning: Implications for the rationality debate?},
  journal = {Behavioral and Brain Sciences},
  volume  = {23},
  number  = {5},
  pages   = {645--665},
  year    = {2000}
}

@inproceedings{ribeiro2020auditing,
  author    = {Ribeiro, Manoel Horta and Ottoni, Raphael and West, Robert and Almeida, Virg{\'i}lio A. F. and Meira Jr., Wagner},
  title     = {Auditing Radicalization Pathways on {YouTube}},
  booktitle = {Proceedings of the 2020 Conference on Fairness, Accountability, and Transparency (FAT* '20)},
  publisher = {ACM},
  pages     = {131--141},
  year      = {2020}
}

@article{hosseinmardi2021examining,
  author  = {Hosseinmardi, Homa and Ghasemian, Amir and Clauset, Aaron and Mobius, Markus and Rothschild, David M. and Watts, Duncan J.},
  title   = {Examining the Consumption of Radical Content on {YouTube}},
  journal = {Proceedings of the National Academy of Sciences},
  volume  = {118},
  number  = {32},
  pages   = {e2101967118},
  year    = {2021}
}

@article{gonzalez2023asymmetric,
  author  = {Gonz{\'a}lez-Bail{\'o}n, Sandra and Lazer, David and Barber{\'a}, Pablo and Zhang, Meiqing and Allcott, Hunt and Brown, Taylor and Crespo-Tenorio, Adriana and Freelon, Deen and Gentzkow, Matthew and Guess, Andrew M. and others},
  title   = {Asymmetric Ideological Segregation in Exposure to Political News on {Facebook}},
  journal = {Science},
  volume  = {381},
  number  = {6656},
  pages   = {392--398},
  year    = {2023}
}

@article{nyhan2023like,
  author  = {Nyhan, Brendan and Settle, Jaime and Thorson, Emily and Wojcieszak, Magdalena and Barber{\'a}, Pablo and Chen, Annie Y. and Allcott, Hunt and Brown, Taylor and Crespo-Tenorio, Adriana and Dimmery, Drew and others},
  title   = {Like-Minded Sources on {Facebook} Are Prevalent but Not Polarizing},
  journal = {Nature},
  volume  = {620},
  number  = {7972},
  pages   = {137--144},
  year    = {2023}
}

@article{levy2021social,
  author  = {Levy, Ro'ee},
  title   = {Social Media, News Consumption, and Polarization: Evidence from a Field Experiment},
  journal = {American Economic Review},
  volume  = {111},
  number  = {3},
  pages   = {831--870},
  year    = {2021}
}

@misc{twitter2023algorithm,
  author       = {{Twitter Engineering}},
  title        = {The Algorithm},
  howpublished = {GitHub: \url{https://github.com/twitter/the-algorithm}},
  note         = {Announcement: \url{https://blog.x.com/en_us/topics/company/2023/a-new-era-of-transparency-for-twitter}},
  year         = {2023}
}

@inproceedings{fernandez2024analysing,
  author    = {Fern{\'a}ndez, Miriam and Bellog{\'i}n, Alejandro and Cantador, Iv{\'a}n},
  title     = {Analysing the Effect of Recommendation Algorithms on the Spread of Misinformation},
  booktitle = {Proceedings of the 16th ACM Web Science Conference (WebSci)},
  year      = {2024},
  pages     = {159--169}
}

@article{pathak2023understanding,
  author  = {Pathak, Royal and Spezzano, Francesca and Pera, Maria Soledad},
  title   = {Understanding the Contribution of Recommendation Algorithms on Misinformation Recommendation and Misinformation Dissemination on Social Networks},
  journal = {ACM Transactions on the Web},
  volume  = {17},
  number  = {4},
  pages   = {35:1--35:26},
  year    = {2023}
}

@article{acemoglu2024model,
  author  = {Acemoglu, Daron and Ozdaglar, Asuman and Siderius, James},
  title   = {A Model of Online Misinformation},
  journal = {Review of Economic Studies},
  volume  = {91},
  number  = {6},
  pages   = {3117--3150},
  year    = {2024}
}

@article{milli2025engagement,
  author  = {Milli, Smitha and Carroll, Micah and Wang, Yike and Pandey, Sashrika and Zhao, Sebastian and Dragan, Anca D.},
  title   = {Engagement, User Satisfaction, and the Amplification of Divisive Content on Social Media},
  journal = {PNAS Nexus},
  volume  = {4},
  number  = {3},
  pages   = {pgaf062},
  year    = {2025}
}

@article{guess2023algorithms,
  author  = {Guess, Andrew M. and Malhotra, Neil and Pan, Jennifer and Barber{\'a}, Pablo and Allcott, Hunt and Brown, Taylor and Crespo-Tenorio, Adriana and Dimmery, Drew and Freelon, Deen and Gentzkow, Matthew and Gonz{\'a}lez-Bail{\'o}n, Sandra and others},
  title   = {How Do Social Media Feed Algorithms Affect Attitudes and Behavior in an Election Campaign?},
  journal = {Science},
  volume  = {381},
  number  = {6656},
  pages   = {398--404},
  year    = {2023}
}

@article{truong2024vulnerabilities,
  author  = {Truong, Bao Tran and Lou, Xiaodan and Flammini, Alessandro and Menczer, Filippo},
  title   = {Quantifying the Vulnerabilities of the Online Public Square to Adversarial Manipulation Tactics},
  journal = {PNAS Nexus},
  volume  = {3},
  number  = {7},
  pages   = {pgae258},
  year    = {2024}
}

@misc{xai2026algorithm,
  author       = {{xAI}},
  title        = {x-algorithm: Algorithm Powering the {For You} Feed on {X}},
  howpublished = {GitHub: \url{https://github.com/xai-org/x-algorithm}},
  year         = {2026},
  note         = {Released January 2026; ranking weights published August 2026. Accessed September 2026}
}

@article{vosoughi2018spread,
  author  = {Vosoughi, Soroush and Roy, Deb and Aral, Sinan},
  title   = {The spread of true and false news online},
  journal = {Science},
  volume  = {359},
  number  = {6380},
  pages   = {1146--1151},
  year    = {2018}
}

@inproceedings{wang2021masknet,
  author    = {Wang, Zhiqiang and She, Qingyun and Zhang, Junlin},
  title     = {{MaskNet}: Introducing Feature-Wise Multiplication to {CTR} Ranking Models by Instance-Guided Mask},
  booktitle = {Proceedings of the 3rd Workshop on Deep Learning Practice for High-Dimensional Sparse Data (DLP-KDD 2021)},
  year      = {2021},
  publisher = {ACM},
  note      = {arXiv:2102.07619}
}

@inproceedings{ye2025auditing,
  author    = {Ye, Jinyi and Luceri, Luca and Ferrara, Emilio},
  title     = {Auditing Political Exposure Bias: Algorithmic Amplification on {Twitter/X} During the 2024 {U.S.} Presidential Election},
  booktitle = {Proceedings of the 2025 ACM Conference on Fairness, Accountability, and Transparency (FAccT '25)},
  year      = {2025},
  pages     = {2349--2362},
  publisher = {ACM}
}

@inproceedings{edelson2021understanding,
  author    = {Edelson, Laura and Nguyen, Minh-Kha and Goldstein, Ian and Goga, Oana and McCoy, Damon and Lauinger, Tobias},
  title     = {Understanding engagement with {U.S.} (mis)information news sources on {Facebook}},
  booktitle = {Proceedings of the 21st ACM Internet Measurement Conference (IMC '21)},
  year      = {2021},
  pages     = {444--463}
}

@article{juul2021comparing,
  author  = {Juul, Jonas L. and Ugander, Johan},
  title   = {Comparing information diffusion mechanisms by matching on cascade size},
  journal = {Proceedings of the National Academy of Sciences},
  volume  = {118},
  number  = {46},
  pages   = {e2100786118},
  year    = {2021}
}

@article{bouchaud2023crowdsourced,
  author  = {Bouchaud, Paul and Chavalarias, David and Panahi, Maziyar},
  title   = {Crowdsourced audit of {Twitter}'s recommender systems},
  journal = {Scientific Reports},
  volume  = {13},
  number  = {1},
  pages   = {16815},
  year    = {2023}
}

@book{miller2007complex,
  author    = {Miller, John H. and Page, Scott E.},
  title     = {Complex Adaptive Systems: An Introduction to Computational Models of Social Life},
  publisher = {Princeton University Press},
  address   = {Princeton, NJ},
  year      = {2007}
}

@article{zhang2020consumption,
  author  = {Zhang, Jingjing and Adomavicius, Gediminas and Gupta, Alok and Ketter, Wolfgang},
  title   = {Consumption and Performance: Understanding Longitudinal Dynamics of Recommender Systems via an Agent-Based Simulation Framework},
  journal = {Information Systems Research},
  volume  = {31},
  number  = {1},
  pages   = {76--101},
  year    = {2020}
}

@article{fleder2009blockbuster,
  author  = {Fleder, Daniel and Hosanagar, Kartik},
  title   = {Blockbuster Culture's Next Rise or Fall: The Impact of Recommender Systems on Sales Diversity},
  journal = {Management Science},
  volume  = {55},
  number  = {5},
  pages   = {697--712},
  year    = {2009}
}

@article{moravec2020appealing,
  author  = {Moravec, Patricia L. and Kim, Antino and Dennis, Alan R.},
  title   = {Appealing to Sense and Sensibility: {System 1} and {System 2} Interventions for Fake News on Social Media},
  journal = {Information Systems Research},
  volume  = {31},
  number  = {3},
  pages   = {987--1006},
  year    = {2020}
}

@article{drolsbach2023diffusion,
  author  = {Drolsbach, Chiara Patricia and Pr{\"o}llochs, Nicolas},
  title   = {Diffusion of Community Fact-Checked Misinformation on {Twitter}},
  journal = {Proceedings of the ACM on Human-Computer Interaction},
  year    = {2023},
  volume  = {7},
  number  = {CSCW2},
  pages   = {1--22}
}

@article{proellochs2023mechanisms,
  author  = {Pr{\"o}llochs, Nicolas and Feuerriegel, Stefan},
  title   = {Mechanisms of True and False Rumor Sharing in Social Media: Collective Intelligence or Herd Behavior?},
  journal = {Proceedings of the ACM on Human-Computer Interaction},
  year    = {2023},
  volume  = {7},
  number  = {CSCW2},
  pages   = {1--38}
}

@article{mcloughlin2024misinformation,
  author  = {McLoughlin, Killian L. and Brady, William J. and Goolsbee, Aden and Kaiser, Ben and Klonick, Kate and Crockett, M. J.},
  title   = {Misinformation Exploits Outrage to Spread Online},
  journal = {Science},
  year    = {2024},
  volume  = {386},
  number  = {6725},
  pages   = {991--996}
}

@article{kim2019says,
  author  = {Kim, Antino and Dennis, Alan R.},
  title   = {Says Who? {T}he Effects of Presentation Format and Source Rating on Fake News in Social Media},
  journal = {MIS Quarterly},
  volume  = {43},
  number  = {3},
  pages   = {1025--1039},
  year    = {2019}
}

@article{moravec2019fake,
  author  = {Moravec, Patricia L. and Minas, Randall K. and Dennis, Alan R.},
  title   = {Fake News on Social Media: People Believe What They Want to Believe When it Makes No Sense At All},
  journal = {MIS Quarterly},
  volume  = {43},
  number  = {4},
  pages   = {1343--1360},
  year    = {2019}
}

@article{kim2019combating,
  author  = {Kim, Antino and Moravec, Patricia L. and Dennis, Alan R.},
  title   = {Combating Fake News on Social Media with Source Ratings: The Effects of User and Expert Reputation Ratings},
  journal = {Journal of Management Information Systems},
  volume  = {36},
  number  = {3},
  pages   = {931--968},
  year    = {2019}
}

@article{oh2013community,
  author  = {Oh, Onook and Agrawal, Manish and Rao, H. Raghav},
  title   = {Community Intelligence and Social Media Services: A Rumor Theoretic Analysis of Tweets During Social Crises},
  journal = {MIS Quarterly},
  volume  = {37},
  number  = {2},
  pages   = {407--426},
  year    = {2013}
}

@article{adomavicius2013do,
  author  = {Adomavicius, Gediminas and Bockstedt, Jesse C. and Curley, Shawn P. and Zhang, Jingjing},
  title   = {Do Recommender Systems Manipulate Consumer Preferences? {A} Study of Anchoring Effects},
  journal = {Information Systems Research},
  volume  = {24},
  number  = {4},
  pages   = {956--975},
  year    = {2013}
}

@article{hosanagar2014will,
  author  = {Hosanagar, Kartik and Fleder, Daniel M. and Lee, Dokyun and Buja, Andreas},
  title   = {Will the Global Village Fracture Into Tribes? {R}ecommender Systems and Their Effects on Consumer Fragmentation},
  journal = {Management Science},
  volume  = {60},
  number  = {4},
  pages   = {805--823},
  year    = {2014}
}

@article{lee2019how,
  author  = {Lee, Dokyun and Hosanagar, Kartik},
  title   = {How Do Recommender Systems Affect Sales Diversity? {A} Cross-Category Investigation via Randomized Field Experiment},
  journal = {Information Systems Research},
  volume  = {30},
  number  = {1},
  pages   = {239--259},
  year    = {2019}
}

@article{susarla2012social,
  author  = {Susarla, Anjana and Oh, Jeong-Ha and Tan, Yong},
  title   = {Social Networks and the Diffusion of User-Generated Content: Evidence from {YouTube}},
  journal = {Information Systems Research},
  volume  = {23},
  number  = {1},
  pages   = {23--41},
  year    = {2012}
}

@article{papanastasiou2020fake,
  author  = {Papanastasiou, Yiangos},
  title   = {Fake News Propagation and Detection: A Sequential Model},
  journal = {Management Science},
  volume  = {66},
  number  = {5},
  pages   = {1826--1846},
  year    = {2020}
}

@article{candogan2020optimal,
  author  = {Candogan, Ozan and Drakopoulos, Kimon},
  title   = {Optimal Signaling of Content Accuracy: Engagement vs. Misinformation},
  journal = {Operations Research},
  volume  = {68},
  number  = {2},
  pages   = {497--515},
  year    = {2020}
}

@article{davis2007developing,
  author  = {Davis, Jason P. and Eisenhardt, Kathleen M. and Bingham, Christopher B.},
  title   = {Developing Theory Through Simulation Methods},
  journal = {Academy of Management Review},
  volume  = {32},
  number  = {2},
  pages   = {480--499},
  year    = {2007}
}

@article{harrison2007simulation,
  author  = {Harrison, J. Richard and Lin, Zhiang and Carroll, Glenn R. and Carley, Kathleen M.},
  title   = {Simulation Modeling in Organizational and Management Research},
  journal = {Academy of Management Review},
  volume  = {32},
  number  = {4},
  pages   = {1229--1245},
  year    = {2007}
}

@article{rand2011agent,
  author  = {Rand, William and Rust, Roland T.},
  title   = {Agent-Based Modeling in Marketing: Guidelines for Rigor},
  journal = {International Journal of Research in Marketing},
  volume  = {28},
  number  = {3},
  pages   = {181--193},
  year    = {2011}
}

@inproceedings{chaney2018algorithmic,
  author    = {Chaney, Allison J. B. and Stewart, Brandon M. and Engelhardt, Barbara E.},
  title     = {How Algorithmic Confounding in Recommendation Systems Increases Homogeneity and Decreases Utility},
  booktitle = {Proceedings of the 12th ACM Conference on Recommender Systems (RecSys '18)},
  pages     = {224--232},
  year      = {2018},
  publisher = {ACM}
}
% ==============================================================

\newpage
\begin{APPENDIX}{Supplementary Material}

\section{Stage 2 Analysis Plan}
\label{app:prereg}

We fixed the analysis plan for Stage~2 before running the confirmatory simulations, and the plan specifies the following five elements. First, for each scoring rule, the \emph{gap} of a metric is its mean over low-credibility tweets minus its mean over high-credibility tweets, $\mathrm{gap} = \mathrm{mean}(\text{metric}\mid\text{low-credibility}) - \mathrm{mean}(\text{metric}\mid\text{high-credibility})$. Second, the \emph{contrast} of a scoring rule is its gap minus the gap under the additive rule, as defined in Section~\ref{sec:propositions}. Third, cascade size is the primary metric, while exposure and time to peak are secondary metrics. Fourth, each scoring rule is simulated $20$ times with different random seeds, and the uncertainty of each contrast is measured by its spread across these replicates. Fifth, the result is classified into one of four outcomes by the confidence intervals of the cascade size contrast, where outcome~\textbf{A} (the interval of the multiplicative form F1 lies entirely below $0$, while the interval of the retuned control includes $0$) supports our claim, outcome~\textbf{B} (the interval of F1 includes $0$) is a null result, outcome~\textbf{C} (both intervals lie entirely below $0$) requires further diagnosis, and outcome~\textbf{D} (the interval of F1 lies entirely above $0$) contradicts our prediction.

Before the confirmatory results were known, we added a dated amendment to the plan with three changes. First, since the spread across replicates captures only random variation between simulator runs and is zero for exposure, which is fully determined by the scores, we added a $1{,}000$-iteration bootstrap that resamples the same cascades under all scoring rules, and we apply the decision rule to this bootstrap interval. Second, we raised the number of replicates from $20$ to $100$. Third, we made exposure a primary metric together with cascade size. Apart from this amendment, the plan is unchanged, and we report the exposure and cascade size results under it in Table~\ref{tab:arch} of Section~\ref{sec:results-stage2}.

\section{Robustness Check Details}
\label{app:robust}

This appendix reports the detailed results of the robustness checks summarized in Section~\ref{sec:robust}. Throughout, a \emph{contrast} is the change in the propagation gap under a scoring rule relative to the additive rule, as defined in Section~\ref{sec:propositions}, so that a negative contrast means that the rule narrows the gap.

\paragraph{Check 1: Aggregation form.}
To check whether the result depends on the specific form of reflective-gated aggregation, we compare each of the three reflective-gated forms and the retuned control with the additive rule, and we report the results in Table~\ref{tab:arch} of the main text. We find that all three reflective-gated forms narrow both the exposure gap and the cascade size gap significantly, while the retuned control moves both gaps in the opposite direction. Among the three forms, the reflective threshold (F3) produces the largest effect on both metrics, and the ratio correction (F2) produces the smallest.

\paragraph{Check 2: Slow/fast partition.}
To check whether the result depends on which engagement types count as fast and which count as slow, we repeat the analysis under four ways of assigning the engagement types to the two classes. We find that all four assignments produce significantly negative contrasts for both exposure and cascade size, with exposure contrasts between $-24.8$ and $-9.4$.

\paragraph{Check 3: Ranker training seed.}
To check whether the result depends on the randomness in training the ranker, we retrain the ranker five times with different random seeds and repeat the analysis with each of them. We find that the exposure contrasts range from $-21.3$ to $-15.5$ and are all significant, while the cascade size contrasts range from $-9.1$ to $-3.8$ and are significantly negative for all five rankers, so the finding does not depend on the randomness in ranker training.

\paragraph{Check 4: Corpus period.}
To check whether the result depends on the period of the corpus, we repeat the analysis on an adjacent week-long slice of the corpus. We find that the exposure contrast is $-16.3$ (s.e.\ $4.82$) on the adjacent slice, which is almost the same as the $-16.14$ (s.e.\ $4.34$) under the baseline configuration in Table~\ref{tab:arch}.

\paragraph{Check 5: Partition and training seed combined.}
To check whether the partition of Check~2 and the training seed of Check~3 interact with each other, we run the full simulation for each of the $4 \times 5 = 20$ combinations of the two. We find that both the exposure contrast and the cascade size contrast are significantly negative in all 20 combinations, so no combination of the two changes the finding.

\paragraph{Check 6: User-pool size.}
To check whether the result depends on the number of simulated users, we repeat the analysis with user pools of 25K, 50K, and 100K users under the default partition and the baseline ranker. We find that all three pool sizes produce the contrasts reported in Table~\ref{tab:arch}, with an exposure contrast of $-16.14$ (s.e.\ $4.34$) and a cascade size contrast of $-3.79$ (s.e.\ $0.94$), since each cascade is simulated from the score and the predicted engagement of its own seed tweet, so the contrasts do not depend on how many users are sampled.

\paragraph{Check 7: Training-corpus size.}
To check whether the result depends on how much data the ranker is trained on, we retrain the ranker on about 2M and then about 5M tweets, and we evaluate each ranker on the same seed tweets as the baseline, so that any change in the contrasts comes from the training data alone. The prediction accuracy of the ranker improves slightly with more data, as the AUC for replies rises from $0.749$ to $0.751$ and $0.755$ as the training set grows from 1M to 2M and 5M tweets. As reported in Table~\ref{tab:check7}, we find that the exposure contrast of the multiplicative form (F1) becomes steadily larger as the training corpus grows ($-16.1$, $-17.2$, and $-19.3$), and that its cascade size contrast, which is already significantly negative under the baseline training, becomes larger with more training data ($-3.8$, $-7.0$, and $-8.1$), while its standard error becomes smaller. In addition, the retuned control continues to move the gap in the opposite direction at every training size we test.

\begin{table}[h]
\centering
\footnotesize
\caption{Contrasts when the ranker is trained on larger corpora (Check 7).}
\label{tab:check7}
\begin{tabular}{lrllc}
\toprule
Training corpus & $n_{\mathrm{train}}$ & Metric & Scoring rule & Estimate \\
\midrule
Baseline & 1.0M & Exposure   & F1 (multiplicative) & $-16.14$$^{***}$ (4.34) \\
Baseline & 1.0M & Cascade size     & F1 (multiplicative) & $-3.79$$^{***}$ (0.94) \\
Enlarged & 2.0M & Exposure   & F1 (multiplicative) & $-17.22$$^{***}$ (3.93) \\
Enlarged & 2.0M & Cascade size     & F1 (multiplicative) & $-7.04$$^{***}$ (0.82) \\
Enlarged & 5.0M & Exposure   & F1 (multiplicative) & $-19.28$$^{***}$ (4.25) \\
Enlarged & 5.0M & Cascade size     & F1 (multiplicative) & $-8.07$$^{***}$ (0.73) \\
Enlarged & 5.0M & Exposure   & Retuned additive (control) & $+1.52$$^{***}$ (0.14) \\
\bottomrule
\end{tabular}
\par\vspace{2pt}{\SingleSpacedXI\footnotesize \textit{Note.} Bootstrap standard errors in parentheses. $^{*}\,p<0.10$, $^{**}\,p<0.05$, $^{***}\,p<0.01$.\par}
\end{table}

\paragraph{Check 8: Ranker design.}
To check whether the result depends on the design of the ranker, we first retrain the ranker on four subsets of its input features, namely the full set of 18 features, the 8 tweet and time features alone, the 10 account features alone, and a minimal set of three features (log followers, log text length, and whether the tweet contains a URL). Since the simulated exposure of a tweet increases with its score under a fixed calibration, we measure the effect of each scoring rule directly on the scores of the seed tweets of the hypothesis test, as the change in the gap between the average scores of low- and high-credibility tweets relative to the additive rule. As reported in Table~\ref{tab:check8}, we find that the multiplicative form (F1) narrows this gap while the retuned control widens it in all four feature subsets, including the minimal three-feature ranker. We then repeat the same procedure with two other types of prediction model in place of MaskNet, namely a small multilayer perceptron and a transformer that uses a two-layer encoder over the 18 feature tokens and is trained in the same way. We find that the same pattern holds for all three model types, and that the transformer, whose prediction accuracy matches that of MaskNet (AUC $0.746$ for replies, $0.845$ for retweets, $0.741$ for likes, and $0.897$ for quotes), produces the largest contrasts of the three ($-0.427$ with s.e.\ $0.0954$ for F1, and $+0.301$ with s.e.\ $0.0339$ for the retuned control). The transformer is closer in design to the model that X deployed in 2026 (Section~\ref{sec:pipeline}), although it is not a reconstruction of that model, since the production model also uses sequences of past user engagement that the USC data do not contain. Therefore, the finding depends neither on the features of the ranker nor on the type of prediction model.

\begin{table}[h]
\centering
\footnotesize
\caption{Score contrasts across ranker feature subsets (Check 8).}
\label{tab:check8}
\begin{tabular}{lrll}
\toprule
Subset & $n_{\mathrm{features}}$ & F1 score contrast & Control score contrast \\
\midrule
Full                & 18 & $-0.354$$^{***}$ $(0.081)$ & $+0.261$$^{***}$ $(0.031)$ \\
Tweet and time only & 8  & $-0.142$$^{***}$ $(0.012)$ & $+0.167$$^{***}$ $(0.011)$ \\
Account only        & 10 & $-0.216$$^{**}$ $(0.088)$ & $+0.096$$^{***}$ $(0.031)$ \\
Minimal             & 3  & $-0.199$$^{***}$ $(0.072)$ & $+0.141$$^{***}$ $(0.030)$ \\
\bottomrule
\end{tabular}
\par\vspace{2pt}{\SingleSpacedXI\footnotesize \textit{Note.} Each score contrast is the change in the gap between the average scores of low- and high-credibility tweets relative to the additive rule. Bootstrap standard errors in parentheses. $^{*}\,p<0.10$, $^{**}\,p<0.05$, $^{***}\,p<0.01$.\par}
\end{table}

\paragraph{Check 9: Gate hyperparameters.}
\label{sec:hyperparam}
While the checks above vary the environment of the experiment, this check varies the hyperparameters of the reflective-gated forms themselves, to see whether the finding depends on how strongly the gate acts. Specifically, we vary the gain $\alpha$ of the multiplicative form (F1) over $\{0.25, 0.5, 1, 2, 4, 8\}$, the gain $\alpha$ of the ratio-correction form (F2) over $\{0.25, 0.5, 1, 2, 4\}$, and the threshold $\theta$ and steepness (scale) of the reflective threshold (F3) over $\theta\in\{0.5,1,1.5,2\}$ and $\mathrm{scale}\in\{0.25,0.5,1\}$, while reusing the same trained ranker, seed tweets, and calibration. As shown in Figure~\ref{fig:hyperparam}, we find that the exposure contrast is significantly negative at all 23 settings, so the direction of the effect never changes. The effect becomes larger as $\alpha$ increases and as $\theta$ increases, and for $\theta \geq 1$ it also becomes larger as the gate becomes steeper, i.e., as the scale becomes smaller. The settings used in the main text, namely $\alpha = 1$ for F1 and F2 and $\theta = 1$ with a scale of $0.5$ for F3, give the values reported in Table~\ref{tab:arch}. It is worth noting that the very large contrasts at $\theta \geq 1.5$ arise because such a gate scales down the scores of most tweets, so that the tweets that clear the threshold receive a much larger exposure allocation, which is why Section~\ref{sec:operating} measures the benefit of the gate by the exposure share of each class at the same total exposure. Table~\ref{tab:hpfull} gives the results for every setting.

\begin{figure}[t]
\centering
\includegraphics[width=0.85\linewidth]{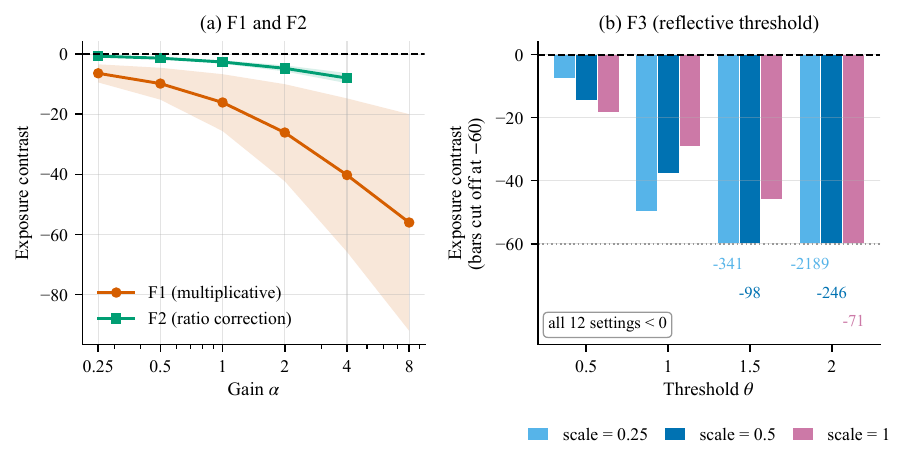}
\caption{Sensitivity to the gate hyperparameters. (a)~Exposure contrast against the gain $\alpha$ of F1 and F2. (b)~Exposure contrast across the threshold and scale settings of F3. Bars below $-60$ are cut off, and their values are printed beneath them.}
\label{fig:hyperparam}
\end{figure}

\begin{table}[h]
\centering
\footnotesize
\caption{Exposure contrast at every hyperparameter setting (Check~9).}
\label{tab:hpfull}
\begin{tabular}{llr}
\toprule
Form & Setting & Estimate \\
\midrule
F1 (multiplicative)  & $\alpha=0.25$ & $-6.38$$^{***}$ (1.54) \\
                              & $\alpha=0.5$  & $-9.81$$^{***}$ (2.70) \\
                              & $\alpha=1.0$  & $-16.14$$^{***}$ (4.34) \\
                              & $\alpha=2.0$  & $-26.09$$^{***}$ (8.24) \\
                              & $\alpha=4.0$  & $-40.23$$^{***}$ (13.0) \\
                              & $\alpha=8.0$  & $-56.01$$^{***}$ (18.4) \\
\midrule
F2 (ratio correction)& $\alpha=0.25$ & $-0.67$$^{***}$ (0.074) \\
                              & $\alpha=0.5$  & $-1.35$$^{***}$ (0.148) \\
                              & $\alpha=1.0$  & $-2.63$$^{***}$ (0.28) \\
                              & $\alpha=2.0$  & $-4.72$$^{***}$ (0.52) \\
                              & $\alpha=4.0$  & $-7.96$$^{***}$ (0.862) \\
\midrule
F3 (reflective threshold)& $\theta=0.5,\,\mathrm{scale}=0.25$ & $-7.67$$^{***}$ (0.921) \\
                              & $\theta=0.5,\,\mathrm{scale}=0.5$  & $-14.54$$^{***}$ (1.97) \\
                              & $\theta=0.5,\,\mathrm{scale}=1.0$  & $-18.43$$^{***}$ (2.90) \\
                              & $\theta=1.0,\,\mathrm{scale}=0.25$ & $-49.60$$^{***}$ (6.61) \\
                              & $\theta=1.0,\,\mathrm{scale}=0.5$  & $-37.83$$^{***}$ (5.05) \\
                              & $\theta=1.0,\,\mathrm{scale}=1.0$  & $-29.23$$^{***}$ (4.76) \\
                              & $\theta=1.5,\,\mathrm{scale}=0.25$ & $-341.0$$^{***}$ (48.7) \\
                              & $\theta=1.5,\,\mathrm{scale}=0.5$  & $-97.55$$^{***}$ (14.4) \\
                              & $\theta=1.5,\,\mathrm{scale}=1.0$  & $-45.80$$^{***}$ (7.81) \\
                              & $\theta=2.0,\,\mathrm{scale}=0.25$ & $-2188.7$$^{***}$ (324) \\
                              & $\theta=2.0,\,\mathrm{scale}=0.5$  & $-245.6$$^{***}$ (39.7) \\
                              & $\theta=2.0,\,\mathrm{scale}=1.0$  & $-70.93$$^{***}$ (12.9) \\
\bottomrule
\end{tabular}
\par\vspace{2pt}{\SingleSpacedXI\footnotesize \textit{Note.} Bootstrap standard errors in parentheses. $^{*}\,p<0.10$, $^{**}\,p<0.05$, $^{***}\,p<0.01$.\par}
\end{table}

\section{Cascade-Simulation and Calibration Procedure}
\label{app:algorithm}

The three parameters of the event model, $(\pi_{\mathrm{active}}, \beta, r)$, are set by matching the following quantities in the observed data,
\begin{align}
\pi_{\mathrm{active}} &\;=\; 1 - P_{\mathrm{obs}}(\mathrm{reply}=0), \label{eq:framework-cal-pi}\\
\beta &\;=\; \frac{\mathbb{E}_{\mathrm{obs}}[\mathrm{reply}\mid \mathrm{reply}>0]}{\mathbb{E}\bigl[S_{\mathrm{rel}}\, p_{\mathrm{reply}}\bigr]}, \label{eq:framework-cal-beta}\\
r &\;=\; \frac{\bar{\mu}_{\mathrm{trim}}^{\,2}}{\,\widehat{\mathrm{Var}}_{\mathrm{trim}} - \bar{\mu}_{\mathrm{trim}}\,}, \label{eq:framework-cal-r}
\end{align}
where $P_{\mathrm{obs}}(\mathrm{reply}=0)$ is the observed share of tweets with no replies, so that the activity gate $\pi_{\mathrm{active}}$ reproduces the share of cascades without replies, $\mathbb{E}_{\mathrm{obs}}[\mathrm{reply}\mid\mathrm{reply}>0]$ is the observed average reply count among tweets with at least one reply, $\mathbb{E}[S_{\mathrm{rel}}\,p_{\mathrm{reply}}]$ is the expected reply intensity over the seed tweets, so that $\beta$ converts this average into the units of the simulator, and $\bar{\mu}_{\mathrm{trim}}$ and $\widehat{\mathrm{Var}}_{\mathrm{trim}}$ are the mean and variance of the non-zero observed reply counts after dropping the largest 5\%. We calibrate these parameters under the additive rule, since the additive rule is the production-style baseline whose outcomes should match the observed magnitudes, and since our design requires that the differences between scoring rules come only from the score-aggregation layer, so the other scoring rules reuse the same parameters without re-fitting. We drop the largest 5\% of the counts before computing the variance, since the raw variance is dominated by a small number of extremely large counts (in the USC corpus, the variance of non-zero reply counts is about $1.5\times10^5$ while their mean is about $38$), which would push the estimate of $r$ to an implausibly small value, produce a U-shaped distribution of simulated counts, and worsen every Stage~1 fit. Dropping these counts gives up some accuracy on the most extreme cascades in exchange for a much better match on typical cascades, which are the ones our mechanism concerns. The estimate of $r$ is kept within $[0.05, 100]$, where values near the upper end make the model close to a Poisson model. Algorithm~\ref{alg:simulation} gives the complete procedure.

\begin{algorithm}[H]
\OneAndAHalfSpacedXI\normalsize
\caption{Cascade simulation with method-of-moments calibration.}
\label{alg:simulation}
\begin{algorithmic}[1]
\Require a corpus part, the trained ranker $f$, the set of scoring rules $\mathcal{R}$, and the number of replicates $M$
\State keep one record per observed author (\texttt{user\_id\_str}), and sample 50{,}000 users stratified by follower-count quintile and paid verification (\texttt{user\_blue})
\State sample the seed tweets, each of which starts a new cascade (for Stage~2, stratified by \texttt{credibility\_label})
\State estimate the hourly posting profile $d$ (24 weights) from the observed posting times
\State $p_{i,k} \gets \sigma\bigl(f_k(\mathbf{x}_i)\bigr)$ for every seed tweet $i$ and engagement objective $k$
\State estimate $(\pi_{\mathrm{active}}, \beta, r)$ under the additive rule \Comment{Eqs.~\eqref{eq:framework-cal-pi}--\eqref{eq:framework-cal-r}}
\State keep $(\pi_{\mathrm{active}}, \beta, r)$ fixed for all scoring rules
\For{scoring rule $\Phi \in \mathcal{R}$}
  \For{replicate $m = 1, \ldots, M$, each with a different random seed}
    \State $S_i \gets \Phi(\mathbf{p}_i)$;\quad $S_{\mathrm{rel},i} \gets S_i / \operatorname{med}(S)$
    \State $\mathit{E}_i \gets \operatorname{clip}(\beta\, S_{\mathrm{rel},i},\, E_{\min},\, E_{\max})$
    \State $\tilde q_{i,t} \propto \exp(-t/\tau)\, d_{(h_i+t) \bmod 24}$, normalized over $t$
    \State $a_i \sim \mathrm{Bernoulli}(\pi_{\mathrm{active}})$
    \State $\lambda_{i,k} \sim \mathrm{Gamma}\bigl(r,\, \mathit{E}_i\, p_{i,k}/r\bigr)$;\quad $n_{i,t,k} \gets a_i \cdot \mathrm{Poisson}\bigl(\lambda_{i,k}\, \tilde q_{i,t}\bigr)$ \Comment{Eqs.~\eqref{eq:framework-mu}--\eqref{eq:framework-counts}}
    \State record the totals of each cascade and its hourly counts $n_{i,t,k}$
  \EndFor
\EndFor
\State combine the replicates for the Stage~1 and Stage~2 analyses
\end{algorithmic}
\end{algorithm}

\FloatBarrier
\section{Calibration Iteration History}
\label{app:calib}

We reached the calibration reported in Section~\ref{sec:results-stage1} in four iterations. Iteration~(i) used a Poisson model with its mean matched to the data, while iteration~(ii) added an activity gate $\pi_{\mathrm{active}}$ matched to the observed share of tweets with no replies, which makes the model a zero-inflated Poisson model. Iteration~(iii) then replaced the Poisson model with a zero-inflated Negative-Binomial model, with the dispersion $r$ estimated from the mean and variance of the non-zero reply counts, and iteration~(iv) keeps this model but estimates $r$ after dropping the largest 5\% of the non-zero counts. Table~\ref{tab:calib} reports how well each iteration matches the observed distributions of root reply counts and of the reflective share of engagement, measured by the Kolmogorov--Smirnov (KS) statistic, where a smaller value indicates a better fit.

\begin{table}[h]
\centering
\footnotesize
\caption{Calibration iterations and their Kolmogorov--Smirnov fit.}
\label{tab:calib}
\begin{tabular}{llp{4.6cm}cc}
\toprule
Iteration & Event model & Estimate of the dispersion $r$ & KS (root reply) & KS (reflective share) \\
\midrule
(i)   & Poisson with the mean matched   & ---                                                     & 0.60  & 0.62  \\
(ii)  & Zero-inflated Poisson           & ---                                                     & 0.090 & 0.097 \\
(iii) & Zero-inflated Negative-Binomial & Mean and variance of the non-zero counts                & 0.165 & 0.150 \\
(iv)  & Zero-inflated Negative-Binomial & \textbf{Mean and variance after dropping the largest 5\% of the non-zero counts} & \textbf{0.057} & \textbf{0.061} \\
\bottomrule
\end{tabular}
\par\vspace{2pt}{\SingleSpacedXI\footnotesize \textit{Note.} KS (root reply) and KS (reflective share) compare the simulated and observed distributions of the root reply count and of the reflective share of engagement.\par}
\end{table}

The improvement from iteration~(iii) to iteration~(iv) comes from the estimate of $r$. Since the raw variance of non-zero reply counts is dominated by a few extremely large counts, the moment estimate $r=\bar\mu^2/(\widehat{\mathrm{Var}}-\bar\mu)$ becomes implausibly small ($r\approx 0.0098$), under which most simulated cascades draw almost no replies while a few draw extremely many, so that iteration~(iii) fits worse than iteration~(ii) on both distributions. Dropping the largest 5\% of the non-zero counts before computing the variance removes the influence of these extreme counts on the estimate while keeping the information about typical cascades, and yields $r=0.318$, which is close to the value that minimizes the KS statistic in a direct grid search ($r\approx 0.5$, KS $0.046$). Trimming the largest observations before estimating a variance is a standard approach in robust statistics to keep a few extreme values from dominating a moment estimate, and among the four iterations, iteration~(iv) gives the best fit.

\section{Data Processing Details}
\label{app:data}

Tweet identifiers (\texttt{id} and \texttt{id\_str}) are stored with all 19 digits, but the conversation identifier (\texttt{conversationId}) and the identifier of the replied-to tweet (\texttt{in\_reply\_to\_status\_id\_str}) are stored in scientific notation, which keeps only about 16 significant digits, so linking replies to the tweets they reply to is approximate for long identifiers. The \texttt{viewCount} and \texttt{user} fields are stored as Python dictionaries written with single quotes and not as JSON, and the \texttt{user} field contains date values written as \texttt{datetime.datetime(...)}, so we first convert these values with a regular expression and then parse the fields with \texttt{ast.literal\_eval}. The legacy verification flag \texttt{user\_verified} is always $0$ in data collected after 2022, so we use the paid-verification flag \texttt{user\_blue} throughout, which is set for $22.9\%$ of tweets and $9.3\%$ of authors in the full corpus. It is worth noting that two properties of the corpus affect the descriptive statistics in Table~\ref{tab:descriptives}. First, about $0.3\%$ of the records are retweets, which carry the engagement counts of the original tweet and not their own, so including them would raise the average retweet count of the corpus from $8.3$ to $20.5$, and we therefore exclude them from the engagement statistics. Second, $453{,}307$ of the $42.85$M records in the release are malformed, almost all of which are truncated records that carry no timestamp, author information, or engagement counts, and another $2{,}096{,}410$ records are repeated copies of a tweet that was collected more than once, and we drop both, which leaves the $40{,}302{,}075$ tweets of Table~\ref{tab:descriptives}.

The ranker uses 18 input features, namely ten account features and eight tweet and time features. The account features are the logarithms of the author's follower count, number of accounts followed, tweets posted, favorites given, list memberships, follower-to-friend ratio, favorites per tweet, tweets per day, and account age, together with the paid-verification flag. The tweet and time features are the logarithm of the text length, whether the tweet is a reply, a quote-tweet, or contains a URL, and the sine and cosine of the posting hour and of the hour of the week.

\end{APPENDIX}

\end{document}